\documentclass[11pt,times]{article}
\usepackage{amsfonts}
\usepackage{amssymb,amsmath,amsthm}
\usepackage{epsfig}
\usepackage{color}
\usepackage{latexsym}
\usepackage{enumerate}
\usepackage{fullpage,setspace}
\usepackage{algorithm}
\usepackage{algorithmic}
\usepackage{bbm}

\usepackage[dvips, paper=letterpaper, top=1in, bottom=.75in, left=1in, right=1in, nohead, includefoot, footskip=.25in]{geometry}
\newcommand{\ind}{\mathbbm{1}}
\usepackage{tweaklist}
\renewcommand{\enumhook}{\setlength{\topsep}{2pt}%
  \setlength{\itemsep}{-2pt}}
\renewcommand{\itemhook}{\setlength{\topsep}{2pt}%
  \setlength{\itemsep}{-2pt}}

\newtheorem{theorem}{Theorem}
\newtheorem{corollary}{Corollary}
\newtheorem{lemma}{Lemma}

\newtheorem{definition}{Definition}

\newtheorem{observation}{Observation}

\newtheorem{remark}{Remark}

\newcommand{\junk}[1]{}
\junk{

}

\newcommand{\costasnote}[1]{{\color{red}{#1}}}
\newcommand{\mattnote}[1]{{\color{blue}{#1}}}

\newenvironment{prevproof}[2]{\noindent {\em {Proof of {#1}~\ref{#2}:}}}{$\Box$\vskip \belowdisplayskip}

\newenvironment{prevtheorem}[2]{\medskip \noindent {\bf {{#1}~\ref{#2}.}} \em}  {\medskip}

\newcommand{\notshow}[1]{}

\definecolor{MyGray}{rgb}{0.8,0.8,0.8}

\begin{document}
\title {On Optimal Multi-Dimensional Mechanism Design}
\author {Constantinos Daskalakis\thanks{Supported by a Sloan Foundation Fellowship and NSF Award CCF-0953960 (CAREER) and CCF-1101491.}\\
EECS, MIT \\
\tt{costis@mit.edu}
\and
S. Matthew Weinberg\thanks{Supported by a NSF Graduate Research Fellowship and a NPSC Graduate Fellowship.}\\
EECS, MIT\\
\tt{smw79@mit.edu}
}
\addtocounter{page}{-1}
\maketitle
\begin{abstract}
\noindent We efficiently solve the {\em optimal multi-dimensional mechanism design problem} for independent bidders with arbitrary demand constraints when either the number of bidders is a constant or the number of items is a constant. In the first setting, we need that each bidder's values for the items are sampled from a possibly correlated,  {\em item-symmetric} distribution, allowing different distributions for each bidder. In the second setting, we allow the values of each bidder for the items to be arbitrarily correlated, but assume that the distribution of bidder types is {{\em bidder-symmetric}}. These symmetric distributions include i.i.d. distributions, as well as many natural correlated distributions. E.g., an item-symmetric distribution can be obtained by taking an arbitrary distribution, and ``forgetting'' the names of items; this could arise when different members of a bidder population have various sorts of correlations among the items, but the items are {``the same''} with respect to a random bidder from the population.

For all $\epsilon>0$, we obtain a computationally efficient additive $\epsilon$-approximation, when the value distributions are bounded, or a multiplicative $(1-\epsilon)$-approximation when the value distributions are unbounded, but satisfy the Monotone Hazard Rate condition, covering a widely studied class of distributions in Economics. Our running time is polynomial in $\max\{\text{\#items,\#bidders}\}$, and {\em not}  the size of the support of the joint distribution of all bidders' values for all items, which is typically exponential in both the number of items and the number of bidders. Our mechanisms are randomized, explicitly price bundles, and in some cases can also accommodate budget constraints.

Our results are enabled by establishing several new tools and structural properties of Bayesian mechanisms. In particular, we provide a {\em symmetrization technique} that turns any truthful mechanism into one that has the same revenue and respects all symmetries in the underlying value distributions. We also prove that item-symmetric mechanisms satisfy a natural {\em strong-monotonicity property} which, unlike cyclic-monotonicity, can be harnessed algorithmically. Finally, we provide a technique that turns any given $\epsilon$-BIC mechanism (i.e. one where incentive constraints are violated by $\epsilon$) into a truly-BIC mechanism at the cost of {$O(\sqrt{\epsilon})$ revenue.} We expect our tools to be used beyond the settings we consider here. Indeed there has already been follow-up research~\cite{CDW,CJ} making use of our tools.
\end{abstract}

\thispagestyle{empty}

\newpage
\section{Introduction} \label{sec:introduction}

How can a seller auction off a set of items to a group of interested buyers to maximize profit? This problem, dubbed the {\em optimal mechanism design} problem, has gained central importance in mathematical Economics over the past decades. The seller could certainly auction off the items sequentially, using her favorite single-item auction, such as the English auction. But this is not always the best idea, as it is easy to find examples where this approach leaves money on the table.~\footnote{A simple example is this: Suppose that an auctioneer is selling a Picasso and a Dali painting and there are two bidders  of which one loves Picasso and does not care about Dali and vice versa. Running a separate English auction for each painting will result in small revenue since there is going to be no serious competition for either painting. But bundling the paintings together will induce competition and drive the auctioneer's revenue higher.} The chief challenge that the auctioneer faces is that the values of the buyers for the items, which determine how much each buyer is willing to pay for each item, is information that is private to the buyers, at least at the onset of the auction. Hence, the mechanism needs to provide the appropriate incentives for the buyers to reveal ``just enough'' information for the optimal revenue to be extracted.

Viewed as an optimization problem, the optimal mechanism design problem is of a rather intricate kind. First, it is  a priori not clear how to evaluate the revenue of an arbitrary mechanism because it is not clear how rational bidders will play. One way to cope with this is to only consider mechanisms where rational bidders are properly incentivized to tell the designer their complete {\em type}, i.e. how much they would value each possible outcome of the mechanism (i.e. each bundle of items they may end up getting). Such mechanisms can be \emph{Incentive Compatible} (IC), where each bidder's strategy is to report a type and the following {\em worst-case guarantee} is met: regardless of the types of the other bidders, it is in the best interest of a bidder to truthfully report her type. Or the mechanism can be \emph{Bayesian Incentive Compatible} (BIC), where it is assumed that the bidders' types come from a known distribution and the following {\em average-case guarantee} is met: in expectation over the other bidders' types, it is in the best interest of a bidder to truthfully report her type, if the other bidders  report truthfully. See Sec~\ref{sec:notation} for formal definitions. We only note here that, under very weak assumptions, restricting attention to IC/BIC mechanisms in the aforementioned settings of without/with prior information over bidders' types  is without loss of generality~\cite{AGTbook}.

But even once it is clear how to evaluate the revenue of a given mechanism, it is not necessarily clear what benchmark to compare it against. For example, it is not hard to see that the {\em social welfare}, i.e. the sum of the values of the buyers for the items they are allocated, is {\em not} the right benchmark to use, as in general one cannot hope to achieve revenue that is within any constant factor of the optimal social welfare: why would a buyer with a large value for an item pay an equally large price to the auctioneer to get it, if there is no competition for this item? Given the lack of a useful revenue benchmark (i.e. one that upper bounds the revenue that one may hope to achieve but is not too large to allow any reasonable approximation), the task of the mechanism designer can only be specified in generic terms as follows: come up with an IC/BIC auction whose revenue is at least as large as the revenue of any other IC/BIC auction.

Finally, even after restricting the search space to IC/BIC auctions and only comparing to the optimal revenue achievable by any IC/BIC auction, it is still easy to show that it is impossible to guarantee any finite approximation if no prior is known over the bidders' types. Instead, many solutions in the literature adopt a Bayesian viewpoint, assuming that a prior does exist and is known to both the auctioneer and the bidders, and targeting the optimal achievable \emph{expected revenue}. Once the leap to the Bayesian setting is made the goal is typically this: {\em Design a BIC, possibly randomized, mechanism whose expected revenue is optimal among all BIC, possibly randomized, mechanisms.}~\footnote{In view of the results of~\cite{BCKW,CMS}, to achieve optimal, or even near-optimal, revenue in correlated settings, or even i.i.d. multi-item settings, we are forced to explore randomized mechanisms.}

\smallskip One of the most celebrated results in this realm is {\em Myerson's optimal auction}~\cite{myerson}, which achieves optimal revenue via an elegant design that spans several important settings. Despite its significance, Myerson's result is limited to the case where bidders are {\em single-dimensional}. In simple terms, this means that each bidder can be characterized by a single number (unknown to the auctioneer), specifying the value of the bidder per item received. This is quite a strong assumption when the items are heterogeneous, so naturally, after Myerson's work, a large body of research has been devoted to the {\em multi-dimensional problem}, i.e. the setting where the bidders may have different values for different items/bundles of items. Even though progress has been made in certain restricted settings, it seems that we are far from an optimal mechanism, generalizing Myerson's result; see survey~\cite{optimal:econ} and its references for work on this problem by Economists.

Algorithmic Game Theory has also studied this problem, with an extra eye on the computational efficiency of mechanisms. Chawla et al.~\cite{CHK} study the case of a single (multidimensional) unit-demand bidder with independent values for the items. They propose an elegant reduction of this problem to Myerson's single-dimensional setting, resulting in  a mechanism that achieves a constant factor approximation to the optimal revenue among all BIC, possibly randomized~\cite{CMS}, mechanisms. For the same problem, Cai and Daskalakis~\cite{CD} recently closed the constant approximation gap against all deterministic mechanisms by obtaining polynomial-time approximation schemes for optimal item-pricing. As for the case of correlated values, it had been known that finding the optimal pricing (deterministic mechanism) is highly inapproximable by~\cite{BK}, although no hardness results are known for randomized mechanisms.    In the multi-bidder setting, Chawla et al.~\cite{CHMS}, Bhattacharya et al.~\cite{BGGM} and recently Alaei~\cite{alaei} obtain constant factor approximations in the case of additive bidders or unit-demand bidders and matroidal constraints on the possible allocations.

While our algorithmic understanding of the optimal mechanism design problem is solid, at least as far as constant factor approximations go, there has been virtually no result in designing computationally efficient revenue-optimal mechanisms for multi-dimensional settings, besides the single-bidder result of~\cite{CD}. In particular, one can argue that the previous approaches~\cite{alaei,BGGM,CHK,CHMS} are inherently limited to constant factor approximations, as ultimately the revenue of these mechanisms is compared against the optimal revenue in a related single-dimensional setting~\cite{CHK,CHMS}, or a convex programming relaxation of the problem~\cite{alaei,BGGM}. Our focus in this work is to fill this important gap in the algorithmic mechanism design literature, i.e. to {\em obtain computationally efficient near-optimal multi-dimensional mechanisms}, {coming $\epsilon$-close to the optimal revenue in polynomial time}, for any desired accuracy $\epsilon>0$. We obtain a Polynomial-Time Approximation Scheme (PTAS) for the following two important cases of the general problem.

\smallskip\noindent  \framebox{
\begin{minipage}{\hsize}
{\bf The BIC $k$-items problem.} Given as input an arbitrary (possibly correlated) distribution $\mathcal{F}$ over valuation vectors for $k$ items, a demand bound $C$, and an integer $m$, the number of bidders, output a BIC mechanism $M$ whose expected revenue is optimal relative to any other, possibly randomized, BIC mechanism, when played by $m$ additive bidders with demand constraint $C$ whose valuation vectors are sampled independently from $\mathcal{F}$.
\end{minipage}}

\smallskip\noindent  \framebox{
\begin{minipage}{\hsize} {\bf The BIC $k$-bidders problem.} Given as input $k$ item-symmetric\footnote{A distribution over $\mathbb{R}^n$ is symmetric if, for all $\vec{v} \in \mathbb{R}^n$, the probability it assigns to $\vec{v}$ is equal to the probability it assigns to any permutation of $\vec{v}$.} distributions $\mathcal{F}_1,\ldots,\mathcal{F}_k$, demand bounds $C_1,\ldots,C_k$ (one for each bidder), and an integer $n$, the number of items, output a BIC mechanism $M$ whose expected revenue is optimal relative to any other, possibly randomized, BIC mechanism, when played by $k$ additive bidders with demand constraints $C_1,\ldots,C_k$ respectively whose valuation vectors for the $n$ items are sampled independently from $\mathcal{F}_1,\ldots,\mathcal{F}_k$.
\end{minipage}}

\smallskip In other words, the problems we study are where either the number of bidders is large, but they come from the same population, i.e. each bidder's value vector is sampled from the same, arbitrary, possibly correlated distribution, or the number of items is large, but each bidder's value distribution is item-symmetric (possibly different for each bidder). While these do not capture the problem of Bayesian mechanism design in its complete generality, they certainly represent important special cases of the general problem and indeed  the first interesting cases for which computationally efficient near-optimal mechanisms have been obtained. Before stating our main result, it is worth noting that:

\medskip \noindent  ~~\begin{minipage}{16.5cm}$\bullet$ When the number of bidders is large, it does not make sense to expect that the auctioneer has a separate prior distribution for the values of each individual bidder for the items. So our assumption in the $k$-items problem that the bidders are drawn from the same population of bidders is a realistic one, and---in our opinion---the practically interesting case of the general problem. Indeed, there are hardly any practical examples of auctions using bidder-specific information (think, e.g., eBay, Sotheby's etc.) A reasonable extension  of our model would be to assume that bidders come from a constant number of different sub-populations of bidders, and that the auctioneer has a prior for each sub-population. Our results extend to this setting.\end{minipage}

\medskip \noindent  ~~\begin{minipage}{16.5cm}$\bullet$ When the number of items is large, it is still hard to imagine that the auctioneer has a distribution for each individual item. In the $k$-bidders problem, we assume that each bidder's value distribution is item-symmetric. This certainly contains the case where each bidder has i.i.d. values for the items, but there are realistic applications where values are correlated, but still item-symmetric. Consider the following scenario: the auctioneer has the same number of Yankees, Red Sox, White Sox, and  Mariners baseball caps to sell.  Each bidder is a fan of one of the four teams and has non-zero value for exactly one of the four kinds of caps, but it is unknown to the auctioneer which kind that is and what the value of the bidder for that kind is. Hence, the values of a random bidder for the caps are certainly non i.i.d., as if the bidder likes a Red Sox cap then she will equally like another, but will have zero value for a Yankees cap. Suppose now that we are willing to make the assumption that all teams have approximately the same number of fans and those fans have statistically the same passion for their team. Then a random bidder's values for the items is drawn from an item-symmetric distribution, or close to one, so we can handle such distributions. In this case too, our techniques still apply if we deviate from the item-symmetric model to models where there is a constant number of types of objects, e.g. caps and jerseys, and symmetries do not permute types, but permute objects within the same type.\end{minipage}

\begin{theorem}\label{thm:additive}(Additive approximation) For all $k$, if $\mathcal{F}$  samples values from $[0,1]^k$ there exists a PTAS with additive error $\epsilon$ for the BIC $k$-items problem. For all $k$, if $\mathcal{F}_i$  samples vectors from $[0,1]^n$, there exists a PTAS with additive error {$\epsilon \cdot \max\{C_i\}$} for the BIC $k$-bidders problem.
\end{theorem}

\begin{remark}\label{rem:additive} Some qualifications on Theorem~\ref{thm:additive} are due.
\begin{itemize}
\item The mechanism output by our PTAS is truly BIC, not $\epsilon$-BIC, and there are no extra assumptions necessary to achieve this.
\item We make no assumptions about the size of the support of $\mathcal{F}_i$ or $\mathcal{F}$, as the runtime of our algorithms {\em does not} depend on the size of the support. This is an important distinction between our work and the literature where it is folklore knowledge that if one is willing to pay computational time polynomial in the size of the support of the value distribution, then the optimal mechanism can be easily computed via an LP (see, e.g.,~\cite{BGGM,BCKW,DFK}). However, exponential size supports are easy to observe. Take, e.g., our $k$-bidders problem and assume that every bidder's value for each of the items is i.i.d. uniform in $\{\$5,\$10\}$. The na\"ive LP based approach would result in time polynomial in $2^n$, while our solution needs time polynomial in $n$.

\item If we are willing to replace BIC by $\epsilon$-BIC (or $\epsilon$-IC) in Thm \ref{thm:additive} and compare our revenue to the best revenue achievable by any BIC (or IC) mechanism, then we can also accommodate budget constraints. The only step of our algorithm that does not respect budgets is the $\epsilon$-BIC to BIC reduction (Thm~\ref{thm:epsilon-BIC to BIC}). For space considerations, we restrict our attention to BIC throughout the main body of the paper and prove the related claims for IC in App~\ref{sec: IC results}.

\item If the value distributions are discrete and every marginal has constant-size support, then our algorithms achieve {\em exactly optimal revenue} in polynomial time, even though the support of such a distribution may well be exponential. For instance, in the example given in the second bullet our algorithm obtains exactly optimal revenue in time polynomial in $n$. In these cases, we can find optimal truly BIC or truly IC mechanisms that also accommodate budget constraints.

\item The mechanisms produced by our techniques satisfy the demand constraints of each bidder in a strong sense (and not in expectation). Moreover, the user of our theorem is free to choose whether they want to satisfy {\em ex-interim individual rationality}, that the expected value of a bidder for the received bundle of items is larger than the expected price she pays, or {\em ex-post individual rationality},  where this constraint is true with probability $1$ (and not just in expectation). We focus the main presentation on producing mechanisms that are ex-interim IR. In App~\ref{app:ex-post IR} we explain the required modification for producing ex-post IR mechanisms {\em without any loss in revenue. }


\item The assumption that $\mathcal{F}_i,\mathcal{F}$ sample from $[0,1]^n$ as opposed to some other bounded set is w.l.o.g. and previous work has made the same assumption~\cite{HKM,HL} on the input distributions.

\end{itemize}
\end{remark}


\noindent One might prefer to assume that the value distributions are not upper bounded, but  satisfy some tail condition, such as the Monotone Hazard Rate condition (see App~\ref{app:MHR}).~\footnote{The class of Monotone Hazard Rate distributions is a family of distributions that is commonly used in Economics applications, and contains such familiar distributions as the Normal, Exponential and Uniform distributions.} Using techniques from~\cite{CD}, we can  extend our theorems to MHR distributions. All the relevant remarks still apply.

\begin{corollary}\label{cor:MHR}(Multiplicative approximation for MHR distributions) For all $k$, if the $k$ marginals of $\mathcal{F}$ all satisfy the MHR condition, there exists a PTAS obtaining at least a $(1-\epsilon)$-fraction of the optimal revenue for the BIC $k$-items problem (whose runtime does not depend on $\mathcal{F}$ or $C$). Likewise, for all $k$, if every marginal of $\mathcal{F}_i$ is MHR for all $i$, there exists a PTAS obtaining at least a $(1-\epsilon)$-fraction of the optimal revenue for the BIC $k$-bidders problem (whose runtime does not depend on $\mathcal{F}_i$ or $C_i$).
\end{corollary}

The rest of the paper is organized as follows: Sec~\ref{sec:notation} provides a  few standard definitions from Mechanism Design. Sec~\ref{sec:overview} gives an overview of our proof of Thm~\ref{thm:additive}, explaining the different components that get into our proof and guiding through the rest of the paper. The rest of the main body and the appendix provide all  technical details. App \ref{app:MHR} provides the proof of Corollary \ref{cor:MHR}.


\section{Preliminaries and notation}\label{sec:notation}

We assume that the seller has a single copy of $n$ (heterogeneous)  items that she wishes to auction to $m$ bidders. Each bidder $i$ has some non-negative value for item $j$ which we denote $v_{ij}$. We can think of bidder $i$'s {\em type} as an $n$-dimensional vector  $\vec{v}_i$, and denote the entire profile of bidders as $\vec{v}$, or sometimes  $(\vec{v}_i~;~\vec{v}_{-i})$ if we want to emphasize its decomposition to the type $\vec{v}_i$ of bidder $i$ and the joint profile $\vec{v}_{-i}$ of all other bidders.  We  denote by $\mathcal{D}$ the distribution from which $\vec{v}$ is sampled. We also denote by $\mathcal{D}_i$ the distribution of types for bidder $i$, and by $\mathcal{D}_{-i}$ the distribution of types for every bidder except $i$. The {\em value of a bidder} with demand $C$ for any subset of items is the sum of her values for her favorite $C$ items in the subset; that is, we assume bidders are {\em additive up to their demand}.

{Since we are shooting for BIC/IC mechanisms, we will only consider (direct revelation) mechanisms} where each bidder's strategy is to report a type. When the reported bidder types are $\vec{v}$, we denote the (possibly randomized) {\em outcome of mechanism $M$} as $M(\vec{v})$. The outcome can be summarized in: the {\em expected price} charged to each bidder (denoted $p_i(\vec{v})$), and a collection of {\em marginal probabilities} $\vec{\phi}(\vec{v}) = (\phi_{ij}(\vec{v}))_{ij}$, where $\phi_{ij}(\vec{v})$ denotes the marginal probability that bidder $i$ receives item $j$.

A collection of marginal probabilities $\vec{\phi}(\vec{v}):=(\phi_{ij}(\vec{v}))_{ij}$
is {\em feasible} iff there exists a consistent with them joint distribution over allocations of items to bidders
so that in addition, with probability $1$, no item is allocated to more than one bidder, and no bidder receives more items than her demand. A straightforward application of the Birkhoff-von Neumann theorem~\cite{JDM} reveals that a sufficient condition for the above to hold is that {\em in expectation} no item is given more than once, and all bidders receive an {\em expected number of items} less than or equal to their demand. Note that this sufficient condition is expressible in terms of the $\phi_{ij}$'s
only. Moreover, under the same conditions, we can efficiently sample a joint distribution with the desired $\phi_{ij}$'s.
 (See App~\ref{app:feasible} for details.)

The outcome of mechanism $M$ restricted to bidder $i$ on input $\vec{v}$ is denoted $M_i(\vec{v}) = (\vec{\phi}_i(\vec{v}), p_i({\vec{v}}))$. Assuming that bidder $i$ is additive (up to her demand) and risk-neutral and that the mechanism is feasible (so in particular it does not violate the bidder's demand constraint) the {\em value} of bidder $i$ for outcome $M_i(\vec{w})$ is just (her expected value) $\vec{v}_i \cdot \vec{\phi}_i(\vec{w})$, while the bidder's {\em utility} for the same outcome is $U(\vec{v}_i,M_i(\vec{w})):=\vec{v}_i \cdot \vec{\phi}_i(\vec{w}) - p_i(\vec{w})$. Such bidders subtracting price from expected value are called {\em quasi-linear.} Moreover, for a given value vector $\vec{v}_i$ for bidder $i$, we write: $\pi_{ij}(\vec{v}_i)=\mathbb{E}_{\vec{v}_{-i}\sim {\cal D}_{-i}}[\phi_{ij}(\vec{v}_i~;~\vec{v}_{-i})]$.

We proceed to formally define incentive compatibility of mechanisms in our notation:
\notshow{
\begin{definition}[\cite{BH,HKM,HL}](BIC/$\epsilon$-BIC/IC/$\epsilon$-IC Mechanism) \label{def:BIC} A mechanism $M$ is called $\epsilon$-BIC iff a bidder with type $\vec{v_i}$ can never expect to gain more than $\epsilon \cdot v_{\max}$ by lying about his type $\vec{v}_i$, i.e. for all $i$, $\vec{v}_i, \vec{w}_i$:
$$\mathbb{E}_{\vec{v}_{-i} \sim {\cal D}_{-i}}\left[U(\vec{v}_i,M_i(\vec{v}))\right] \ge \mathbb{E}_{\vec{v}_{-i} \sim {\cal D}_{-i}}\left[ U(\vec{v}_i,M_i(\vec{w}_i~;~\vec{v}_{-i})) \right] - \epsilon \cdot v_{\max},$$
where $v_{\max}$ is the maximum possible value of any bidder for any item in the support of the value distribution. Similarly, $M$ is called $\epsilon$-IC iff  for all $i$, $\vec{v}_i, \vec{w}_i, \vec{v}_{-i}$: $U(\vec{v}_i,M_i(\vec{v})) \ge  U(\vec{v}_i,M_i(\vec{w}_i~;~\vec{v}_{-i}))  - \epsilon \cdot v_{\max}$. A mechanism is called BIC iff it is $0$-BIC and IC iff it is $0$-IC.
~\footnote{Clearly, the definition of $\epsilon$-BIC is meaningless for unbounded distributions, for $\epsilon>0$. However, for such distributions, we shall never define/use $\epsilon$-BIC mechanisms. We may truncate and discretize an unbounded distribution and design $\epsilon$-BIC mechanisms for the resulting distribution. But then we always convert this mechanism to a $0$-BIC mechanism for the original distribution. We will never misuse this definition and try to claim that we have an $\epsilon$-BIC mechanism for an unbounded distribution, for $\epsilon>0$.}
\end{definition}

}
{
\begin{definition}(BIC/$\epsilon$-BIC/IC/$\epsilon$-IC Mechanism)\label{def:BIC} A mechanism $M$ is called $\epsilon$-BIC iff the following inequality holds for all $i,\vec{v}_i,\vec{w}_i$:
$$\mathbb{E}_{\vec{v}_{-i} \sim {\cal D}_{-i}}\left[U(\vec{v}_i,M_i(\vec{v}))\right] \ge \mathbb{E}_{\vec{v}_{-i} \sim {\cal D}_{-i}}\left[ U(\vec{v}_i,M_i(\vec{w}_i~;~\vec{v}_{-i})) \right] - \epsilon v_{\max} \cdot \sum_j \pi_{ij}(\vec{w}_i),$$
where $v_{\max}$ is the maximum possible value of any bidder for any item in the support of the value distribution. In other words, $M$ is $\epsilon$-BIC iff when a bidder lies by reporting $\vec{w}_i$ instead of $\vec{v}_i$, they do not expect to gain more than $\epsilon v_{\max}$ times the expected number of items that $\vec{w}_i$ receives. Similarly, $M$ is called $\epsilon$-IC iff  for all $i$, $\vec{v}_i, \vec{w}_i, \vec{v}_{-i}$: $U(\vec{v}_i,M_i(\vec{v})) \ge  U(\vec{v}_i,M_i(\vec{w}_i~;~\vec{v}_{-i}))  - \epsilon v_{\max} \cdot \sum_j \phi_{ij}(\vec{w}_i~;~\vec{v}_{-i})$. A mechanism is called BIC iff it is $0$-BIC and IC iff it is $0$-IC.~\footnote{Any feasible mechanism that we call $\epsilon$-BIC, respectively $\epsilon$-IC, by our definition is certainly an $\epsilon \cdot \max\{C_i\}$-BIC, respectively $\epsilon \cdot \max\{C_i\}$-IC, mechanism by the more standard definition, which omits the factors $\sum_j \pi_{ij}(\vec{w}_i)$, respectively $\sum_j \phi_{ij}(\vec{w}_i~;~\vec{v}_{-i})$, from the incentive error. We only include these factors here for convenience.}
\end{definition}
}

\noindent In our proof of Thm~\ref{thm:additive} throughout this paper we assume that $v_{\max}=1$. If $v_{\max}<1$, we can scale the value distribution so that this condition is satisfied.
{We also define individual rationality of BIC/$\epsilon$-BIC mechanisms:

\begin{definition}
A BIC/$\epsilon$-BIC mechanism $M$ is called {\em ex-interim individually rational (ex-interim IR)} iff for all $i$, $\vec{v}_i$:
$$\mathbb{E}_{\vec{v}_{-i} \sim {\cal D}_{-i}}\left[U(\vec{v}_i,M_i(\vec{v}))\right] \ge 0.$$
It is called {\em ex-post individually rational (ex-post IR)} iff for all $i$, $\vec{v}_i$ and $\vec{v}_{-i}$,  $U(\vec{v}_i,M_i(\vec{v})) \ge 0$ with probability $1$ (over the randomness in the mechanism).
\end{definition}
\noindent While we focus the main presentation to obtaining ex-interim IR mechanisms, in Appendix~\ref{app:ex-post IR} we describe how without any loss in revenue we can turn these mechanisms into ex-post IR.}

\smallskip For a mechanism $M$, we denote by $R^M(\mathcal{D})$ the expected revenue of the mechanism when bidders sampled from $\mathcal{D}$ play $M$ truthfully. We also let $R^{OPT}(\mathcal{D})$ (resp. $R^{OPT}_\epsilon(\mathcal{D})$) denote the maximum possible expected revenue attainable by any BIC (resp. $\epsilon$-BIC) mechanism when bidders are sampled from ${\cal D}$ and play truthfully. For all cases we consider, these terms are well-defined.

We state and prove our results assuming that we can exactly sample from all input distributions efficiently and exactly evaluate their {cumulative distribution} functions. Our results still hold {\em even if we only have oracle access to sample from the input distributions}, as this is sufficient for us to approximately evaluate the {cumulative} functions to within the right accuracy in polynomial time (by making use of our symmetry and discretization tools, described in the next section). The approximation error on evaluating the cumulative functions is absorbed into  loss in revenue. See discussion in App \ref{app:input model}.

Finally, we denote by $S_m, S_n$ the symmetric groups over the sets $[m]:=\{1,\ldots,m\}$ and $[n]$ respectively. Moreover, for $\sigma =(\sigma_1,\sigma_2) \in S_m \times S_n$, we assume that $\sigma$ maps element $(i,j) \in [m] \times [n]$ to $\sigma(i,j):=(\sigma_1(i), \sigma_2(j))$. We extend this definition to map a value vector $\vec{v}=(v_{ij})_{i\in [m], j \in [n]}$ to the vector $\vec{w}$ such that $\vec{w}_{\sigma(i,j)}=\vec{v}_{ij}$, for all $i,j$. Likewise, if ${\cal D}$ is a value distribution, $\sigma({\cal D})$ is the distribution that first samples $\vec{v}$ from $\cal D$ and then outputs $\sigma(\vec{v})$.

\section{Overview of our Approach}  \label{sec:techniques} \label{sec:overview}

{\bf A Na\"ive LP Formulation.} Let ${\cal D}$ be the distribution of all bidders' values for all items (supported on a subset of $\mathbb{R}^{m \times n}$, where $m$ is the number of bidders and $n$ is the number of items). For a mechanism design problem with unit-demand bidders whose values are distributed according to ${\cal D}$, it is folklore knowledge how to write a linear programming relaxation of size polynomial in $|{\rm supp}({\cal D})|$ optimizing revenue. The relaxation keeps track of the (marginal) probability $\phi_{ij}(\vec{v}) \in [0,1]$ that item $j$ is given to bidder $i$ if the bidders' values are $\vec{v}$, and enforces feasibility constraints (no item is given more than once in expectation, no bidder gets more than one item in expectation), incentive compatibility constraints (in expectation over the other bidders' values for the items, no bidder has incentive to misreport her values for the items, if the other bidders don't), while optimizing the expected revenue of the mechanism. Notice that all constraints and the objective function can be written in terms of the marginals $\phi_{ij}$. Moreover, using the Birkhoff-von Neumann decomposition theorem, it is possible to convert the solution of this LP to a mechanism that has the same revenue and satisfies the feasibility constraints strongly (i.e. not in expectation, but prob. $1$). We give the details of the linear program in App~\ref{app:LP}, and also describe how to generalize this LP to incorporate demand and budget constraints.

Despite its general applicability, the na\"ive LP formulation has a major drawback in that $|{\rm supp}(\mathcal{D})|$ could in general be infinite, and when it is finite it is usually exponential in both $m$ and $n$. For the settings we consider, this is always the case. For example, in the very simple setting where ${\cal D}$ samples each value i.i.d. uniformly from $\{\$5, \$10\}$, the support of the distribution becomes $2^{m \times n}$. Such support size is obviously prohibitive if we plan to employ the na\"ive LP formulation to optimize revenue.

\smallskip \noindent {\bf A Comparison to Myerson's Setting.} {\em What enables succinct and computationally efficient mechanisms in the single-item setting of Myerson?} Indeed, the curse of dimensionality discussed above arises even when there is a single item to sell; e.g., if every bidder's distribution has support $2$ and the bidders are independent, then the number of different bidder profiles is already $2^m$. What drives Myerson's result is the realization that there is structure in {a BIC mechanism} coming in the form of {\em monotonicity}: for all $i$, for all $v_{i1} \ge v_{i1}'$: $\mathbb{E}_{\vec{v}_{-i}}(\phi_{i1}(v_{i1}~;~\vec{v}_{-i})) \ge \mathbb{E}_{\vec{v}_{-i}}(\phi_{i1}(v'_{i1}~;~\vec{v}_{-i})),$ i.e. the expected probability that bidder $i$ gets the single item for sale in the auction increases with the value of bidder $i$, where the expectation is taken over the other bidders' values. Unfortunately, such crisp monotonicity property of {BIC mechanisms} fails to hold if there are multiple items, and even if it were present it would still not be sufficient in itself to reduce the size of the na\"ive LP to a manageable size.

 {\em So what next?} We argued earlier that the symmetric distributions considered in the BIC $k$-items and the BIC $k$-bidders problems are very natural cases of the general optimal mechanism design problem.  We argue next that they are natural for another reason: they enable enough structure for (i) the optimal mechanism to have small description complexity, instead of being an unusable, exponentially long list of what the mechanism ought to do for every input value vector $\vec{v}$; and (ii) the succinct solution to be  efficiently computable, bypassing the exponentially large na\"ive LP. Our structural results are discussed in the following paragraphs. The first is enabled by exploiting randomization to transfer symmetries from the value distribution to the optimal mechanism. The second is enabled by proving a {\em strong-monotonicity property} of all BIC mechanisms. Our notion of monotonicity is more powerful than the notion of {\em cyclic-monotonicity}, which holds more generally but can't be exploited algorithmically. Together our structural results bring to light how the item- and bidder-symmetric settings are mathematically more elegant than general settings with no {apparent} structure.

%

\smallskip \noindent {\bf Structural Result 1:}~{\em The Interplay Between Symmetries and Randomization.} Since the inception of Game Theory scientists were interested in the implications of symmetries in the structure of equilibria~\cite{GKT,BvN,N}. In his seminal paper~\cite{N}, Nash showed a rather interesting structural result, informally reading as follows: ``If a game has any symmetry, there exists a Nash equilibrium satisfying that symmetry." Indeed, something even more powerful is true: ``There always exists a Nash equilibrium that simultaneously satisfies all symmetries that the game may have.''

\notshow{
In his seminal paper~\cite{N}, Nash showed that the use of randomness in players' strategies enables the existence of an equilibrium in every game. In the same paper, Nash showed a less well-known, albeit deep, structural implication of randomness: if all players of the game are identical (i.e. they have identical strategy sets and payoff functions), there exists a symmetric equilibrium in the game, i.e. one in which every player uses the same randomized strategy. In fact, something stronger is true: there always exists a Nash equilibrium respecting any symmetry that the game has. For example, suppose that only two players of a multi-player game are identical (i.e. swapping their names does not affect any player's payoff); then there exists a Nash equilibrium in which these two players use the same mixed strategy. If two actions are interchangeable (i.e. replacing all occurrences of one strategy with the other in a strategy profile does not affect any player's payoff), then there exists a Nash equilibrium in which every player plays these two actions with equal probability. In a game that has both of the above symmetries, there exists a Nash equilibrium satisfying both of the symmetries described above, etc. }

Inspired by Nash's symmetry result, albeit in our different setting, we show a similar structural property of {\bf randomized} mechanisms.~\footnote{We emphasize `randomized', since none of the symmetries we describe holds for deterministic optimal mechanisms.} Our structural result is rather general, applying to settings beyond those addressed in Thm~\ref{thm:additive}, and even beyond MHR or regular distributions. The following theorem holds for \emph{any} (arbitrarily correlated) joint distribution $\mathcal{D}$.

\begin{theorem} \label{thm:symmetries} \label{thm:symmetry}  Let ${\cal D}$ be the distribution of bidders' values for the items (supported on a subset of $\mathbb{R}^{m \times n}$). Let also ${\cal S} \subseteq S_m \times S_n$ be an arbitrary set such that ${\cal D} \equiv \sigma(\cal D)$, for all $\sigma \in {\cal S}$; that is, assume that $\cal D$ is invariant under all permutations in $\cal S$. Then any BIC mechanism $M$  can be symmetrized into a mechanism $M'$ that respects all symmetries in $\cal S$ without any loss in revenue. I.E. for all bid vectors $\overrightarrow{v}$ the behavior of $M'$ under $\overrightarrow{v}$ and $\sigma(\overrightarrow{v})$ is identical (up to permutation by $\sigma$) for all $\sigma \in {\cal S}$. The same result holds if we replace BIC with $\epsilon$-BIC, IC, or $\epsilon$-IC.
\end{theorem}


While we postpone further discussion of this theorem and what it means for $M$ to behave ``identically'' to Sec~\ref{sec:symmetries}, we give a quick example to illustrate the symmetries that randomization enables in the optimal mechanism. Consider a single bidder and two items. Her value for each item is drawn i.i.d. from the uniform distribution on $\{4,5\}$. It is easy to see that the only optimal deterministic mechanism assigns price $4$ to one item and $5$ to the other. However, there is an optimal randomized mechanism that offers each item at price $4\frac{1}{2}$, and the uniform lottery ($1/2$ chance of getting item $1$, $1/2$ chance of getting item $2$) at price $4$. While item $1$ and item $2$ need to be priced differently in the deterministic mechanism to achieve optimal revenue, they can be treated identically in the optimal randomized mechanism. Thm \ref{thm:symmetries} applies in an extremely general setting: distributions can be continuous with arbitrary support and correlation, bidders can have budgets and demands, we could be maximizing social welfare instead of revenue, etc. 

\smallskip \noindent {\bf Structural Result 2:} {\em Strong-Monotonicity.} Even though the na\"ive LP formulation is not computationally efficient, Thm \ref{thm:symmetries} certifies the existence of a compact solution for the cases we consider. This solution lies in the subspace of $\mathbb{R}^{m \times n}$ spanned by the symmetries induced by ${\cal D}$. Still Thm~\ref{thm:symmetries} does not inform us how to locate such a symmetric optimal solution. Indeed, the symmetry of the optimal solution is not a priori capable in itself to decrease the size of our na\"ive LP to a manageable one. For this purpose we establish a strong monotonicity property of {item-symmetric BIC} mechanisms (an item-symmetric mechanism is one that respects every item symmetry; see Sec \ref{sec:symmetries} for a definition).

\begin{theorem}\label{thm:monotone} If ${\cal D}$ is  item-symmetric, every item-symmetric BIC mechanism is {\em strongly monotone}:
$$\text{for all bidders $i$, and items $j, j'$:}~v_{ij} \ge v_{ij'} \implies \mathbb{E}_{\vec{v}_{-i}}(\phi_{ij}(\vec{v})) \ge  \mathbb{E}_{\vec{v}_{-i}}(\phi_{ij'}(\vec{v})).$$
\end{theorem}
\noindent i.e., if $i$ likes item $j$ more than item $j'$, her expected probability (over the other bidders' values) of getting item $j$ is higher. We give an analogous monotonicity property of IC mechanisms in Appendix~\ref{sec: IC results}.

\notshow{While it is a priori not clear that we should be able to find such a solution efficiently, we can indeed write a polynomial-size linear program that outputs a succinct representation of a symmetric solution. The LP looks similar to the original, except it keeps a single representative valuation vector $\overrightarrow{v}$ per equivalence class (under the symmetries of $\cal D$) of valuation vectors, and computes a collection of price-lottery pairs for these representatives only. Clearly, this approach alone is problematic as the truthfulness constraints on representatives are not sufficient to guarantee truthfulness of the full-fledged mechanism: roughly speaking, the smaller LP accounts only for potential deviations to other representative valuation vectors, and not permutations thereof.  }

\smallskip \noindent {\bf From $\epsilon$- to truly-BIC.} Exploiting the aforementioned structural theorems we are able to efficiently compute {\em exactly optimal mechanisms} for value distributions ${\cal D}$ whose marginals on every item have constant-size support.  (${\cal D}$ itself can easily have exponentially-large support if, e.g., the items are independent.) To adapt our solution to continuous distributions {or distributions whose marginals have non-constant support}, we attempt the obvious rounding idea that changes $\mathcal{D}$ by rounding all values sampled from $\mathcal{D}$ down to the nearest multiple of some accuracy $\epsilon$, and solves the problem on the resulting distribution ${\cal D}_{\epsilon}$. While we can argue that the optimal BIC mechanism for ${\cal D}_{\epsilon}$ is also approximately optimal for ${\cal D}$, we need to also give up on the incentive compatibility constraints, resulting in an approximately-BIC mechanism where bidders may have an incentive to misreport their values, but the incentive to misreport is always smaller than some function of $\epsilon$. A natural approach to eliminate those incentives to misreport is to appropriately discount the prices for items or bundles of items charged by the mechanism computed for ${\cal D}_{\epsilon}$, generalizing the single-bidder rounding idea attributed to Nisan in~\cite{CHK}. Unfortunately, this approach fails to work in the multi-bidder settings, destroying both the revenue and truthfulness. Simply put, even though the discounts encourage bidders to choose more expensive options, these choices affect not only the price they pay us, but the prices paid by other bidders as well as the incentives of other bidders. Once we start rounding the prices, we could completely destroy any truthfulness the original mechanism had, leaving us with no guarantees on revenue.

Our approach is  entirely different, comprising a non-trivial extension of the main technique of~\cite{HKM}. We run simultaneous VCG auctions, one per bidder, where each bidder competes with make-believe replicas of himself, whose values are drawn from the same value distribution where his own values are drawn from. The goods for sale in these per-bidder VCG auctions are replicas of the bidder drawn from the modified distribution ${\cal D}_{\epsilon}$. These replicas are called surrogates. The intention is that the surrogates bought by the bidders in the per-bidder VCG auction will compete with each other in the optimal mechanism ${ M}$ designed for the modified distribution ${\cal D}_{\epsilon}$. Accordingly, the {value} of a bidder for a surrogate is the expected value of the bidder for the items that the surrogate is expected to win in ${ M}$ {\em minus} the price the surrogate is expected to pay. {This is exactly our approach,} except we modify mechanism $M$ to discount all these prices by a factor of $1-O(\epsilon)$. This is necessary to argue that bidders choose to purchase a surrogate with high probability, as otherwise we cannot hope to make good revenue. There are several technical ideas coming into the design and analysis of our two-phase auction (surrogate sale, surrogate competition). We describe these ideas in detail in Sec~\ref{sec:true PTAS}, emphasizing several important complications departing from the setting of~\cite{HKM}. Importantly, the approach of \cite{HKM} is brute force in $|\text{supp}(\mathcal{D}_i)|$. While this is okay for $k$-items, this takes exponential time for $k$-bidders. In addition to showing the following theorem, we show how to make use of Thm \ref{thm:monotone} to get the reduction to run in polynomial time in both settings.

\begin{theorem}\label{thm:epsilon-BIC to BIC}
{Consider a generic setting with $n$ items and $m$ bidders who are additive up to some capacity.} Let ${\cal D}:=\times_i {\cal D}_i$ and ${\cal D}':=\times_i {\cal D}'_i$ be product distributions, sampling every bidder independently from $[0,1]^n$. Suppose that, for all $i$, ${\cal D}'_i$ samples vectors whose coordinates are integer multiples of some $\delta \in (0,1)$ and that ${\cal D}_i$ and ${\cal D}'_i$ can be coupled so that, with probability $1$, a value vector $\vec{v}_i$ sampled from ${\cal D}_i$ and a value vector $\vec{v}'_i$ sampled from ${\cal D}'_i$ satisfy that $v_{ij} \ge v_{ij}' \ge v_{ij}-\delta, \forall j$. Then, for all {$\eta, \epsilon >0$}, any $\epsilon$-BIC mechanism $M_1$ for ${\cal D}'$ can be transformed into a BIC mechanism $M_2$ for ${\cal D}$ such that $R^{M_2}({\cal D}) \ge (1-\eta) \cdot R^{M_1}({\cal D}') - \frac{\epsilon+2\delta}{\eta}T$, where $T$ is the maximum number of items that can be awarded by a feasible mechanism. {Furthermore, if $\mathcal{D}$ and $\mathcal{D}'$ are both valid inputs to the BIC $k$-bidders or $k$-items problem, the transformation runs in time polynomial in $n$ and $m$. Moreover, for the BIC $k$-items problem, $T = k$ and, for the BIC $k$-bidders problem, $T \leq k \max_i C_i$, where $C_i$ is the demand of bidder $i$.}
\end{theorem}

\noindent Figure~\ref{fig:structure} shows how the various components discussed above interact with each other to prove Theorem~\ref{thm:additive}. The proof of Corollary~\ref{cor:MHR} is given in Appendix~\ref{app:MHR}.

  \begin{figure}[h!]
  \centering
  \includegraphics[width=12cm]{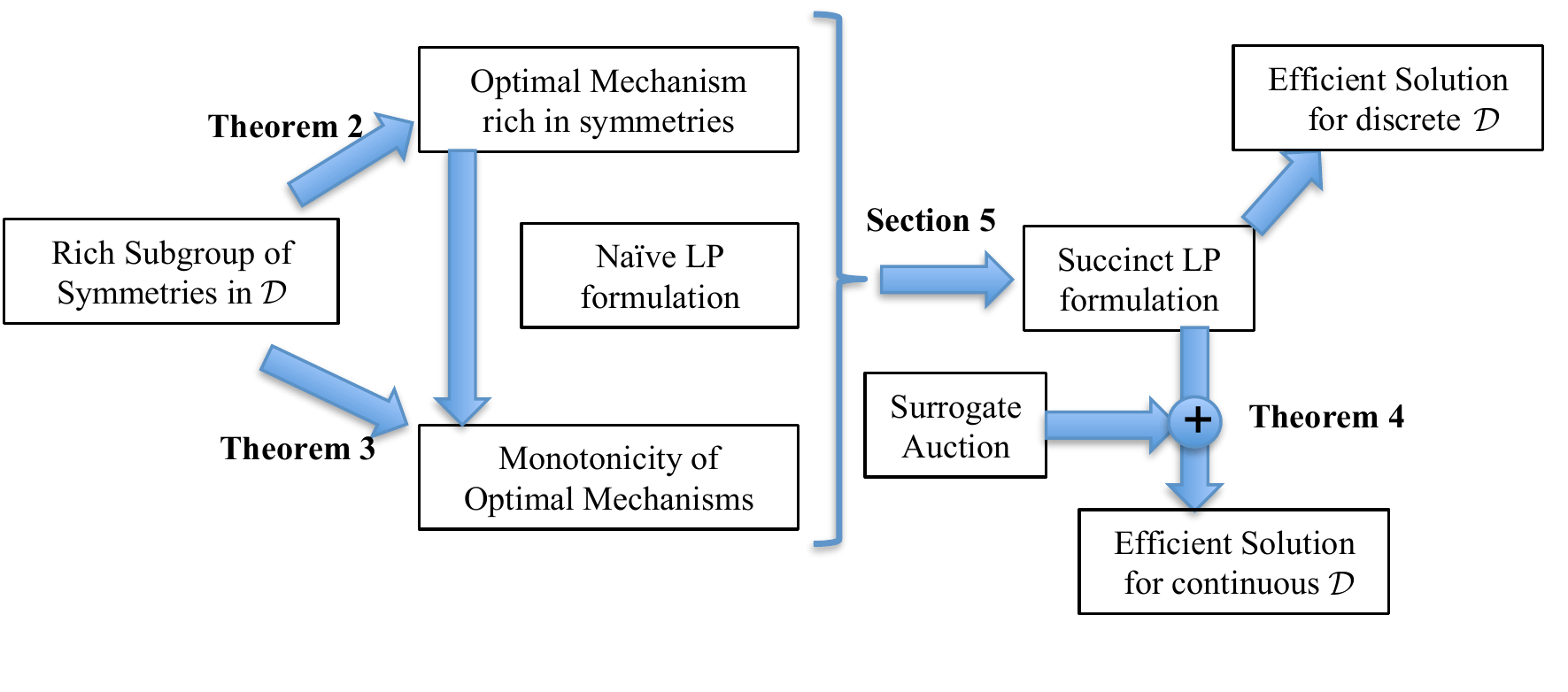}\\
    \caption{Our Proof Structure}\label{fig:structure}
  \end{figure}

\section{Symmetry Theorem}\label{sec:symmetries}
We provide the necessary definitions to understand exactly what our symmetry result is claiming.

\begin{definition}(Symmetry in a Distribution) We say that a distribution $\mathcal{D}$ has symmetry $\sigma \in S_m \times S_n$ if, for all $\vec{v} \in \mathbb{R}^{m\times n}$, $Pr_{\cal D}[\vec{v}] = Pr_{\cal D}[\sigma(\vec{v})]$. We also write ${\cal D} \equiv \sigma({\cal D})$.
\end{definition}

\begin{definition}(Symmetry in a Mechanism) We say that a mechanism respects symmetry $\sigma \in S_m \times S_n$ if, for all $\vec{v} \in \mathbb{R}^{m\times n}$, $M(\sigma(\vec{v})) = \sigma(M(\vec{v}))$.
\end{definition}

\begin{definition}(Permutation of a Mechanism) For any $\sigma \in S_m \times S_n$, and any mechanism $M$, define the mechanism $\sigma(M)$ as $[\sigma(M)](\vec{v}) = \sigma(M(\sigma^{-1}(\vec{v})))$.
\end{definition}
\notshow{\mattnote{I want to remove this. It's taking up a lot of space and I don't think it's necessary.}
While the definition for symmetry of a distribution seems very natural, the definition for symmetry of a mechanism and permuting a mechanism may seem odd. We might instead expect to see $M(\sigma(\vec{v})) = M(\vec{v})$ and $\sigma(M)(\vec{v}) = \sigma(M(\vec{v}))$. These latter definitions are in fact very unnatural. Here is a single-bidder $2$-item example to illustrate why. Suppose that $M(\langle 1,0 \rangle)$ gives item $1$ with probability $1/2$ at price $1/4$. Let $\sigma$ swap items $1$ and $2$. Then by these alternative definitions, for $M$ to respect $\sigma$, it must also have $M(\langle 0,1 \rangle)$ give item $1$ with probability $1/2$ at price $1/4$, an outcome the bidder is not willing to purchase. A more natural definition would be to ask that $M(\langle 0,1 \rangle)$ give item $2$ with probability $1/2$ at price $1/4$, exactly what our definition asks. Likewise, by this alternative definition we would have defined $[\sigma(M)](\langle 1,0 \rangle)$ to yield item $2$ with probability $1/2$ at price $1/4$, again an outcome the bidder is not willing to purchase. A more natural definition would be to have $[\sigma(M)](\langle 0,1 \rangle)$ yield item $2$ with probability $1/2$ at price $1/4$, again what our definition asks.} We proceed to state our symmetry theorem; its proof can be found in App~\ref{app:symmetries}.

\begin{prevtheorem}{Theorem}{thm:symmetry}{\bf (Restated from Sec~\ref{sec:overview})} For all $\mathcal{D}$, any BIC (respectively IC, $\epsilon$-IC, $\epsilon$-BIC) mechanism $M$ can be symmetrized into {a BIC (respectively IC,  $\epsilon$-IC, $\epsilon$-BIC) mechanism} $M'$ such that, for all $\sigma  \in S_m \times S_n$, if $\mathcal{D}$ has symmetry $\sigma$, $M'$ respects $\sigma$, and $R^M(\mathcal{D}) = R^{M'}(\mathcal{D})$.
\end{prevtheorem}

We note that Thm~\ref{thm:symmetry} is an extremely general theorem. $\mathcal{D}$ can have arbitrary correlation between bidders or items, and can be continuous. \notshow{In addition, the theorem still holds if we allow for demand or budgets constraints and other generalizations of the problem.} One might wonder why we had to restrict our theorem to symmetries in $S_m \times S_n$ and not arbitary permutations of the set $[m]\times [n]$. In fact, after reading through our proof, one can see that the same inequalities that make symmetries in $S_m \times S_n$ work also hold for symmetries in $S_{[m]\times [n]}$. However, the mechanism resulting from our proof is not a feasible one, since our transformation can violate feasibility constraints for symmetries $\sigma \notin S_m \times S_n$.

We also emphasize a subtle property of our symmetrizing transformation: the transformation takes as input a set of symmetries satisfied by ${\cal D}$ and a mechanism, and symmetrizes the mechanism so that it satisfies all symmetries in the given set of symmetries. Our transformation {\em does not work if the given set of symmetries is not a subgroup.} Luckily the maximal subset of symmetries in $S_m \times S_n$ satisfied by a value distribution is {\em always a subgroup,} and this enables our result.

\section{Optimal Symmetric Mechanisms for Discrete Distributions}\label{sec:LP}
In this section, we solve the following problem: ``Given a distribution $\mathcal{D}$ with constant support per dimension and a subgroup of symmetries $S \subseteq S_m \times S_n$ satisfied by ${\cal D}$, find a BIC mechanism $M$ that respects all symmetries in $S$ and maximizes $R^M(\mathcal{D})$.'' By Thm~\ref{thm:symmetry}, such $M$ will in fact be optimal with respect to all mechanisms. Intuitively, optimizing over symmetric mechanisms should require less work than over general mechanisms, since we should be able to exploit the symmetry constraints in our optimization. Indeed, suppose that every bidder can report $c$ different values for each item, where $c$ is some absolute constant. Then the na\"ive LP of Section~\ref{sec:overview}/App~\ref{app:LP} has size polynomial in $c^{mn}$, where  $m,n$ are the number of bidders and items respectively. In Sec~\ref{subsec:bidder symmetries} we give a simple observation that reduces the number of variables and constraints of this LP for any given $S$.  This observation in itself is sufficient to provide an efficient solution to the BIC $k$-items problem (in our constant-support-per-dimension setting), but falls short from solving the BIC $k$-bidders problem. For the latter, we need another structural consequence of symmetry, which comes in the form of a {\em strong-monotonicity} property satisfied by all symmetric BIC mechanisms. Strong-monotonicity and symmetry together enable us to obtain an efficient solution to the BIC $k$-bidders problem in Sec~\ref{subsec:item symmetries} (still for our constant-support-per-dimension setting).  We explicitly write the LPs that find the optimal BIC mechanism. Simply tacking on a {$-\epsilon\cdot \sum_j \pi_{ij}(\vec{w}_i)$} to the correct side of the BIC constraints yields an LP to find the optimal $\epsilon$-BIC mechanism for any $\epsilon$. Efficiently solving non-constant/infinite supports per dimension is postponed to Sec~\ref{sec:Algorithms}.

\subsection{Reducing the LP size for any $S$, and Solving the discrete BIC $k$-items Problem}\label{subsec:bidder symmetries}
We provide an LP formulation that works for any $S$. Our LP is the same as the na\"ive LP of Figure~\ref{fig:naive LP}, except we drop some constraints of that LP and modify its objective function as follows. Since our mechanism needs to respect every symmetry in $S$, it must satisfy
$$\phi_{ij}(\vec{v}) = \phi_{\sigma(i,j)}(\sigma(\vec{v})), \forall i,j,\vec{v},\sigma \in S\text{ and }p_i(\vec{v}) = p_{\sigma(i)}(\sigma(\vec{v})), \forall i,\vec{v}, \sigma \in S.$$
Therefore, if we define an equivalence relation by saying that $\vec{v} \sim_S \sigma(\vec{v})$, for all $\sigma \in S$, we only need to keep variables $\phi_{ij}(\vec{v}),p_i(\vec{v})$ for a single representative from each equivalence class. We can then use the above equalities to substitute for all non-representative $\vec{v}$'s into the na\"ive LP. This will cause some constraints to become duplicates. If we let $E$ denote the set of representatives, then we are left with the LP of Figure~\ref{fig:succinct k-items LP} in App \ref{app:succinct LPs}, after removing duplicates. In parentheses at the end of each type of variable/constraint is the number of {distinct} variables/constraints {of that type}.

\begin{lemma} \label{lem:simple reduction works for k-items} The LP of Fig.~\ref{fig:succinct k-items LP} in App~\ref{app:succinct LPs} has polynomial size for the BIC $k$-items problem, if the support of every marginal of the value distribution is an absolute constant.
\end{lemma}

\subsection{Strong-Monotonicity, and Solving the discrete BIC $k$-bidders Problem}\label{subsec:item symmetries}

Unfortunately, the reduction of the previous section is \emph{not} strong enough to make the LP polynomial in the number of items $n$, even if $S$ contains all item permutations and there is a constant number of bidders. This is because a bidder can deviate to an exponential number $c^n$ of types, and our LP needs to maintain an exponential number of BIC constraints. To remedy this, we prove that every item-symmetric {BIC} mechanism for bidders sampled from an item-symmetric distribution satisfies a natural monotonicity property:

\begin{definition}(Strong-Monotonicity of a BIC mechanism) A BIC or $\epsilon$-BIC mechanism is said to be {\em strongly monotone} if for all $i,j,j'$, $v_{ij} \geq v_{ij'} \Rightarrow \pi_{ij}(\vec{v_i}) \geq \pi_{ij'}(\vec{v_i})$. That is, bidders expect to receive their favorite items more often.
\end{definition}
\begin{prevtheorem}{Theorem}{thm:monotone}{\bf (Restated from Sec~\ref{sec:overview})} If $M$ is BIC and $\mathcal{D}$ and $M$ are both item-symmetric, then $M$ is strongly monotone. If $M$ is $\epsilon$-BIC and $\mathcal{D}$ and $M$ are both item-symmetric, there exists a {$\epsilon$-BIC} mechanism of the same expected revenue that is strongly monotone.
\end{prevtheorem}\\
\noindent The proof of Thm~\ref{thm:monotone} can be found in App~\ref{app:proof LP}. We note again that our notion of strong-monotonicity is different than the notion of cyclic-monotonicity that holds more generally, but is not sufficient for obtaining efficient algorithms. Instead strong-monotonicity suffices due to the following:
\begin{observation}\label{obs:monotone} When playing {an item-symmetric,} strongly monotone BIC mechanism, bidder $\vec{v_i}$ has no incentive to report any $\vec{w_i}$ with $w_{ij} > w_{ij'}$ unless $v_{ij} \geq v_{ij'}$.
\end{observation}
\begin{lemma} \label{lem: succinct LP for k-bidders} There exists a polynomial-size LP for the BIC $k$-bidders problem, if the support of every marginal of the value distribution is an absolute constant. The LP is shown in Fig.~\ref{fig:succinct k-bidders LP} of App \ref{app:succinct LPs}.
\end{lemma}

{We note that Theorem \ref{thm:monotone} is also true for IC and $\epsilon$-IC mechanisms with the appropriate definition of strong-monotonicity. The definition and proof are given in Appendix \ref{sec: IC results}.}

\section{Efficient Mechanisms for General Distributions}\label{sec:Algorithms}
We use the results of Sec \ref{sec:LP} to prove Thm~\ref{thm:additive}. First, it is not hard to see that discretizing the value distribution to multiples of $\delta$, for sufficiently small $\delta = \delta(\epsilon)$, and applying Lemmas~\ref{lem:simple reduction works for k-items} and~\ref{lem: succinct LP for k-bidders} yields an algorithm for computing an $\epsilon$-BIC $\epsilon$-optimal mechanism for the $k$-items and $k$-bidders problems. The resulting technical difficulty is turning these mechanisms into being $0$-BIC. To do this, we employ a non-trivial modification of the construction in~\cite{HKM} to improve the truthfulness of the mechanism at the cost of a small amount of revenue. We present our construction and its challenges in Sec~\ref{sec:true PTAS}.
\subsection{A Warmup: $\epsilon$-Truthful Near-Optimal Mechanisms} \label{sec:epsilon-truthful}

\paragraph{Discretization:} Let ${\cal D}$ be a valid input to the BIC $k$-items or the BIC $k$-bidders problem. For each $i$, create a new distribution $\mathcal{D}'_i$ that first samples a bidder from $\mathcal{D}_i$, and rounds every value down to the nearest multiple of $\delta$. Let ${\cal D}'$ be the product distribution of all ${\cal D}'_i$'s. {Let also $T$ denote the maximum number of items that can be awarded by a feasible mechanism.} We show the following lemma whose proof can be found in App~\ref{app:bi-criterion}.

\begin{lemma}\label{lem:deltaIC}
For all $\delta$, let $M'$ be the optimal {$\delta$}-BIC mechanism for ${\cal D}'$. Then $R^{M'}(\mathcal{D}') \geq R^{OPT}(\mathcal{D}) - {\delta T}$. Moreover, let $M$ be the mechanism that on input $\vec{v}$ rounds every $v_{ij}$ down to the nearest multiple $v_{ij}'$ of $\delta$ and implements the outcome $M'(\vec{v}')$. Then $M$ is {$2\delta$}-BIC for bidders sampled from $\mathcal{D}$, and has revenue at least $R^{OPT}(\mathcal{D})-{\delta T}$.
\end{lemma}

Now notice that our algorithms of Sec \ref{sec:LP} allow us to find  an optimal {$\delta$}-BIC mechanism $M'$ for $\mathcal{D}'$. So an application of Lemma~\ref{lem:deltaIC} allows us to obtain a {$2\delta$}-BIC mechanism for ${\cal D}$ whose revenue is at least $R^{OPT}(\mathcal{D})-{\delta T}$.

\subsection{Truthful Near-Optimal Mechanisms: Proof of Theorems~\ref{thm:epsilon-BIC to BIC} and~\ref{thm:additive}}\label{sec:true PTAS}

{We start this section with describing our $\epsilon$-BIC to BIC transformation result (Thm~\ref{thm:epsilon-BIC to BIC}) arguing that it can be implemented efficiently in the BIC $k$-items and $k$-bidders settings.~{\footnote{{We will explicitly describe the transformation for the BIC $k$-items and $k$-bidders settings. For an arbitrary setting where $m$ bidders sample their valuation vectors for $n$ items independently (but not necessarily identically) from $[0,1]^n$ (allowing correlation among items), simply employ the BIC $k$-items transformation, replacing $k$ with $n$.}}} Combining our transformation  with the results of the previous sections we obtain Thm~\ref{thm:additive} in the end of this section. Our transformation is inspired by~\cite{HKM}, but has several important differences. We explicitly describe our transformation, point out the key differences between our setting and that considered in \cite{HKM}, and outline the proof of correctness, postponing the complete proof of Theorem~\ref{thm:epsilon-BIC to BIC} to Appendix~\ref{app:true PTAS}.}

\paragraph{Algorithm Phase $1$: Surrogate Sale}
\begin{enumerate}
\item Recall from the statement of Theorem~\ref{thm:epsilon-BIC to BIC} that ${\cal D}'$ samples values that are integer multiples of $\delta$ and that ${\cal D}$ and ${\cal D}'$ can be coupled so that, whenever we have $\vec{v}$ sampled from $\mathcal{D}$ and $\vec{v}'$ sampled from $\mathcal{D}'$, we have $v_{ij} \geq v'_{ij} \geq v_{ij}-\delta$, for all $i,j$. Moreover, $M_1$ is an $\epsilon$-BIC mechanism for $\mathcal{D}'$, for some $\epsilon$.

\item {Modify $M_1$ to multiply all prices it charges by a factor of $(1-\eta)$. Call $M$ the  mechanism resulting from this modification. Interpret the $\eta$-fraction of the prices given back as rebates.}
\item For each bidder $i$, create $r-1$ \emph{replicas} sampled i.i.d. from $\mathcal{D}_i$ and $r$ \emph{surrogates} sampled i.i.d. from $\mathcal{D}'_i$. Use {$r = ({\eta \over \delta})^2 \cdot m^2 \cdot \hat{\beta}$, where $\hat{\beta} = ({1 \over \delta}+1)^k$, for the $k$-items transformation, and $\hat{\beta}=(n+1)^{1/\delta +1}$, for the $k$-bidders transformation.}
\item Ask each bidder to report $\vec{v_i}$. For $k$-bidders only: Fix a permutation $\sigma$ such that ${v_{i\sigma^(j)} \geq v_{i\sigma^(j+1)}}, \forall j$. For each  surrogate and replica $\vec{w_i}$, permute $\vec{w_i}$ into $\vec{w_i}'$ satisfying ${w'_{i\sigma^(j)} \geq w'_{i\sigma^(j+1)},\forall j}$.
\item Create a weighted bipartite graph with replicas (and bidder $i$) on the left and surrogates on the right. The weight of an edge between replica (or bidder $i$) with type $\vec{r_i}$ and surrogate of type $\vec{s_i}$ is $\vec{r_i}$'s utility for the expected outcome of $\vec{s_i}$ when playing $M$ (where the expectation is taken over the randomness of $M$ and of the other bidders assuming they are sampled from $\mathcal{D}'_{-i}$).
\item Compute the VCG matching and prices. If a replica (or bidder $i$) is unmatched in the VCG matching, add an edge to a random unmatched surrogate. The surrogate selected for bidder $i$ is whoever she is matched to.
\end{enumerate}
\paragraph{Algorithm Phase $2$: Surrogate Competition}
\begin{enumerate}
\item Let $\vec{s_i}$ denote the surrogate chosen to represent bidder $i$ in phase one, and let $\vec{s}$ denote the entire surrogate profile. Have the surrogates $\vec{s}$ play $M$.
\item If bidder $i$ was matched to their surrogate through VCG, charge them the VCG price and award them $M_i(\vec{s})$. (Recall that this has both an allocation and a price component; the price  is added onto the VCG price.) If bidder $i$ was matched to a random surrogate after VCG, award them nothing and charge them nothing.
\end{enumerate}

{There are several differences between our transformation and that of~\cite{HKM}. First, observe that, because ${\cal D'}$ and $M_1$ are explicitly given as input to our transformation (via an exact sampling oracle from ${\cal D}'_i$ and explicitly specifying the outcome awarded to every type $\vec{v}_i$ sampled from ${\cal D}'_i$, for all $i$), we do not have to worry about approximation issues in calculating the edge weights of our VCG auctions in Phase 1. Second, in \cite{HKM}, the surrogates are taking part in an algorithm rather than playing a mechanism, and every replica has non-negative {value} for the outcome of an algorithm because there are no prices charged. Here, however, replicas may have negative {value} for the outcome of a mechanism because there are prices charged. Therefore, some edges may have negative weights, and the VCG matching may not be perfect. We have modified $M$ to give rebates (phase $1$, step $2$) so that the VCG matching cannot be far from perfect, and show that we do not lose too much revenue from unmatched bidders. Finally, in the $k$-bidders problem, the vanilla approach  that does not permute sampled replicas and surrogates (like we do in phase $1$, step $4$ of our reduction) would require exponentially many replicas and surrogates to preserve revenue. To maintain the computational efficiency of our reduction, we resort to sampling only polynomially many replicas/surrogates and permuting them according to the permutation induced by the bidder's reported values. This may seem like it is giving a bidder control over the distribution of replicas and surrogates sampled for her. We show, exploiting the monotonicity results of Sec \ref{sec:LP}, that our construction  is still BIC despite our permuting the replicas and surrogates. We overview the main steps of the proof of Thm~\ref{thm:epsilon-BIC to BIC} and give its complete proof in App~\ref{app:true PTAS}. We conclude this section with the proof of Thm~\ref{thm:additive}.

\medskip \begin{prevproof}{Theorem}{thm:additive}
Choose $\mathcal{D}'$ to be the distribution that samples from ${\cal D}$ and  rounds  every $v_{ij}$ down to the nearest multiple of $\delta$. Let then $M_1$  be the optimal $\delta$-BIC mechanism for $\mathcal{D}'$ as computed by the algorithms of Section~\ref{sec:LP}. By Lemma~\ref{lem:deltaIC},  $R^{M_1}(\mathcal{D}') \geq R^{OPT}(\mathcal{D}) - \delta T$. Applying Thm~\ref{thm:epsilon-BIC to BIC} we obtain a BIC mechanism $M_2$ such that

{\begin{align}
R^{M_2}({\cal D}) &\ge (1-\eta) \cdot R^{M_1}({\cal D}') - \frac{3\delta}{\eta}T \\
&\ge R^{OPT}(\mathcal{D}) - \eta \cdot R^{OPT}(\mathcal{D}) - (1-\eta) \delta T - \frac{3\delta}{\eta}T. \label{eq:grand 1}
\end{align} 

Notice that  $R^{OPT}(\mathcal{D}) {\le T}$. Hence, choosing  $\eta=\epsilon$ and $\delta = \epsilon^2$, \eqref{eq:grand 1} gives
\begin{align}
R^{M_2}({\cal D}) &\ge  R^{OPT}(\mathcal{D}) - O(\epsilon \cdot k)~~~~\text{(for $k$-items)}; \text{and}\\
R^{M_2}({\cal D}) &\ge  R^{OPT}(\mathcal{D}) - O(\epsilon \cdot \sum_i C_i)~~~~\text{(for $k$-bidders)}.
\end{align}
The proof of Theorem~\ref{thm:additive} is concluded by noticing that $\sum_i C_i \le k \max_i C_i$ and $k$ is an absolute constant.
}\end{prevproof}

\appendix


\section{Input Distribution Model}\label{app:input model}

We discuss two  models for accessing a value distribution ${\cal D}$, and explain what modifications are necessary, if any, to our algorithms to work with each model:

\begin{itemize}

\item {\bf Exact Access:} We are given access to a sampling oracle as well as an oracle that exactly integrates the pdf of the distribution over a specified region.

\item {\bf  Sample-Only Access:} We are given access to a sampling oracle and nothing else.

\end{itemize}

The presentation of the paper focuses on the first model. In this case, we can exactly evaluate probabilities of events without any special care. If we have sample-only access to the distribution, we need to be a bit more careful, and proceed to sketch what modifications are necessary to the LP of Section~\ref{sec:LP} and the reduction of Section~\ref{sec:Algorithms} to obtain our results. For this discussion to be meaningful, the reader should be familiar with the notation of Sections~\ref{sec:LP} and~\ref{sec:Algorithms} and their appendices. Let $\{E_{\ell}\}_{\ell}$ denote the partition of $[0,1]^{nm}$ induced by rounding and symmetries. That is, let $\vec{v} \sim \vec{w}$ if there exists a symmetry $\sigma$ that $\mathcal{D}$ has, and there exist integers $c_{11},\ldots,c_{mn}$ such that $c_{ij}\delta > v_{\sigma(i,j)},w_{ij} \geq (c_{ij}-1)\delta$, for all $i,j$;  let then $\{E_{\ell}\}_{\ell}$ denote the partition of $[0,1]^{mn}$ induced by this equivalence relation. In the settings we consider, we have counted at most $\max\{n,m\}^{{(1/\delta+1)}^{\min\{n,m\}}}$ different $E_{\ell}$'s. Hence, we can take $\tilde{O}\left(1/\zeta^2 \cdot \max\{n,m\}^{{(1/\delta+1)}^{\min\{n,m\}}}\right)$ samples from the oracle to simultaneously estimate all probabilities $\{\Pr[\vec{v} \in E_{\ell}]\}_{\ell}$ to within additive accuracy $\zeta$, with high probability. W.l.o.g. we can assume that the estimated probabilities $\{\widehat{\Pr}[\vec{v} \in E_{\ell}]\}_{\ell}$ sum to exactly $1$ (as our estimator is just the histogram of the equivalence classes in which the samples from ${\cal D}$ have fallen). Given these estimates, let ${\cal D}_{\delta,\zeta}$ be the discrete distribution that samples an event $E_{\ell}$ with probability $\widehat{\Pr}[\vec{v} \in E_{\ell}]$, and then outputs an arbitrary vector whose coordinates are all integer multiples of $\delta$ within $E_{\ell}$, after having permuted that vector according to a random $\sigma \in {\cal S}$ where ${\cal S} \subseteq S_m \times S_n$ is the set of symmetries satisfied by ${\cal D}$. Given ${\cal D}_{\delta,\zeta}$ we modify our algorithm as follows. First, we apply the algorithms of Section~\ref{sec:LP} to the distribution ${\cal D}_{\delta,\zeta}$ (that is known explicitly) to compute the optimal mechanism $M_1$ for this distribution. Then we skip the rounding of Section~\ref{sec:epsilon-truthful}; and most importantly, in our reduction of Section~\ref{sec:true PTAS} we make sure to sample surrogates from ${\cal D}_{\delta,\zeta}$ and \emph{not} sample from $\mathcal{D}$ and then round down to multiples of $\delta$. This is because we need the expected outcomes of $M_1$ when played by surrogates sampled by $\mathcal{D}_{\delta,\zeta}$ to be exactly as they were computed by the LPs of Section~\ref{sec:LP}. On the other hand, we still sample replicas directly from $\mathcal{D}$. The error coming from using ${\cal D}_{\delta,\zeta}$ instead of the discretized (to multiples of $\delta$) version of ${\cal D}$ in the reduction of Section~\ref{sec:true PTAS} can be folded into the revenue approximation error. With the appropriate choice of $\zeta$ we can still obtain Theorem~\ref{thm:additive} and Corollary~\ref{cor:MHR}. We skip further details.


\section{A Na\"ive Linear Programming Formulation} \label{app:LP}

Let $\cal D$ be the joint distribution of all bidders' values for all items; this distribution is supported on some subset of $\mathbb{R}^{m \times n}$, where $m$ is the number of bidders and $n$ is the number of items. It has long been known that finding the optimal randomized BIC (or IC) mechanism is merely solving a linear program of size polynomial in $n,m,$ and $|{\rm supp}(\mathcal{D})|$.~\footnote{Such formulation is folklore and appears among other places in~\cite{BCKW,BGGM,DFK}.} The simple linear program to find the optimal BIC mechanism for a distribution with finite support $\mathcal{D}$ where bidder $i$ has demand $C_i$ and budget $B_i$ is shown in Figure \ref{fig:naive LP}. Set $B_i = +\infty$ for bidders with no budget constraints.

\begin{figure}[h]
\colorbox{MyGray}{
\begin{minipage}{\textwidth} {
\noindent\textbf{Variables:}
\begin{itemize}
\item $\phi_{ij}(\vec{v})$, for all bidders $i$, items $j$, and bidder profiles $\vec{v} \in {\rm supp}(\mathcal{D})$, denoting the probability that bidder $i$ receives item $j$ when the bidder profile is $\vec{v}$.
\item $\pi_{ij}(\vec{v}_i)$, for all bidders $i$, items $j$, and types $\vec{v}_i \in {\rm supp}(\mathcal{D}_i)$ for bidder $i$,  denoting the expected probability that bidder $i$ receives item $j$ when reporting type $\vec{v}_i$ (where the expectation is taken with respect to the types of the other bidders).
\item $p_i(\vec{v})$, for all bidders $i$, and bidder profiles $\vec{v} \in {\rm supp}(\mathcal{D})$, denoting the expected price that bidder $i$ pays when the bidder profile is $\vec{v}$.
\item $q_i(\vec{v}_i)$, for all bidders $i$, and types $\vec{v}_i \in {\rm supp}(\mathcal{D}_i)$  for bidder $i$, denoting the expected price that bidder $i$ pays when reporting type $\vec{v}_i$ (where the expectation is taken with respect to the types of the other bidders and the randomness in the mechanism).
\end{itemize}
\textbf{Constraints:}
\begin{itemize}
\item $\pi_{ij}(\vec{v}_i) = \sum_{\vec{w}|\vec{w}_i = \vec{v}_i} Pr[\vec{w} | \vec{w}_i = \vec{v}_i]\phi_{ij}(\vec{w})$, for all $i,j,\vec{v}_i$, guaranteeing that $\pi_{ij}(\vec{v}_i)$ is computed correctly.
\item $q_i(\vec{v}_i) = \sum_{\vec{w}|\vec{w}_i = \vec{v}_i} Pr[\vec{w}|\vec{w}_i = \vec{v}_i]p_i(\vec{w})$, for all $i,\vec{v}_i$, guaranteeing that $q_i(\vec{v}_i)$ is computed correctly.
\item $0 \leq \phi_{ij}(\vec{v}) \leq 1$, for all $i,j,\vec{v}$, guaranteeing that each $\phi_{ij}(\vec{v})$ is a probability.
\item $\sum_i \phi_{ij}(\vec{v}) \leq 1$, for all $j,\vec{v}$, guaranteeing that no item is awarded more than once in expectation.
\item $\sum_j \phi_{ij}(\vec{v}) \leq C_i$, for all $i,\vec{v}$, guaranteeing that no bidder is awarded more than $C_i$ items in expectation.
\item $p_i(\vec{v}) \leq B_i$, for all $i,\vec{v}$, guaranteeing that no bidder ever pays more than $B_i$ in expectation.
\item $\sum_{j} v_{ij}\pi_{ij}(\vec{v}_i) - q_i(\vec{v}_i) \geq 0$, for all $i,\vec{v}_i$, guaranteeing that the mechanism is ex-interim individually rational (IR).
\item $\sum_j v_{ij}\pi_{ij}(\vec{v}_i) - q_i(\vec{v}_i) \geq \sum_j v_{ij}\pi_{ij}(\vec{v}'_i) - q_i(\vec{v}'_i)$, for all $i,\vec{v}_i$, and $\vec{v}'_i$, guaranteeing that the mechanism is BIC.
\end{itemize}
\textbf{Maximizing:}
\begin{itemize}
\item $\sum_{i,\vec{v}} p_i(\vec{v})Pr[\vec{v}]$, the expected revenue.\\
\end{itemize}}
\end{minipage}} \caption{Na\"ive LP for bidders with demand and budget constraints.}\label{fig:naive LP}
\end{figure}

A simple application of the Birkhoff-Von Neumann theorem tells us that as long that marginals $\phi_{ij}(\vec{v})$ satisfy the demand and supply constraints in expectation, then we can find in polynomial time a distribution over allocations that satisfies the demand and supply constraints deterministically and induces these marginals. In addition, a nice trick allows us to switch between ex-post IR and ex-interim IR with no hurt in the value of the LP. These methods are described in Appendix \ref{app:feasible} and \ref{app:ex-post IR}, respectively.

\section{Feasible Randomized Allocations}\label{app:feasible}
Here, we show how to efficiently turn the $\phi$s of a mechanism into an actual randomized outcome. We start with unit-demand bidders (i.e. $C_i=1$ for all $i$) and explain what modifications are necessary for non-unit demand bidders. We note that this procedure was also used in \cite{DFK}, but we include it again here for completeness. Given a collection $\{\phi_{ij}\}_{i,j}$ we explicitly find a distribution over feasible deterministic outcomes (a deterministic outcome assigns each item to some bidder, or possibly the trash) that assigns bidder $i$ item $j$ with probability $\phi_{ij}$ using the Birkhoff-Von Neumann decomposition of a doubly stochastic matrix. To do this, we put the $\phi_{ij}$'s into a matrix, $\Phi$. We observe that $\Phi$ is almost doubly-stochastic, except that the sums of rows or columns can be less than $1$, and $\Phi$ isn't square. We can change $\Phi$ into $\Phi'$ that is doubly-stochastic in the following way: First, add dummy items or dummy bidders to make $\Phi$ square. Next, step through each entry of $\Phi$ one by one and increase $\Phi_{ij}$ as much as possible without making row $i$ or column $j$ sum to greater than $1$. Now we have a $\Phi'$ that is doubly stochastic.

Next, we run a constructive algorithm for the Birkhoff-Von Neumann theorem (\cite{JDM}) to decompose $\Phi'$ into the weighted sum of at most $(\max\{m,n\})^2$ permutation matrices in ${\rm poly}(\max\{m,n\})$ time. Now our sampling scheme is as follows. Pick a permutation matrix with probability equal to it's weight in the decomposition of $\Phi'$, and call this matrix $P$. If $P_{ij} = 1$, then give bidder $i$ item $j$ with probability $\Phi_{ij}/\Phi'_{ij}$.

For any $i,j$, let's explicitly compute the probability that bidder $i$ gets item $j$ in this sampling procedure. The probability that $P_{ij} = 1$ is exactly $\Phi'_{ij}$. And the probability that bidder $i$ gets item $j$ is exactly $\Phi_{ij}/\Phi'_{ij}$ times the probability that $P_{ij} = 1$, which is exactly $\Phi_{ij}$.

\paragraph{Handling Non-Unit Demand Bidders.} If bidder $i$ has demand $C_i$ (instead of $1$), we can replace her in the matrix $\Phi$ with $\min(n,C_i)$ copies (where $n$ is the number of items), each receiving at most one item in expectation. Then we can run the same decomposition and give each bidder all the items awarded to her copies. This solution still runs in polynomial time, and always awards each bidder at most $C_i$ items.

\section{Modifications Required for Ex-Post Individually Rational Mechanisms} \label{app:ex-post IR}
Here, we describe how to turn an ex-interim IR mechanism into an ex-post IR mechanism. If $M$ is ex-interim IR, then we just have $\sum_j v_{ij}\pi_{ij}(\vec{v}_i) \geq q_i(\vec{v}_i)$ for all $i,\vec{v}_i$. Our modification is this: Let $c_i(\vec{v}_i) := q_i(\vec{v}_i)/(\sum_j v_{ij}\pi_{ij}(\vec{v}_i))$, for some specific $i, \vec{v}_i$. Then whenever bidder $i$ receives bundle $J$ when his bid was $\vec{v}_i$, charge him $\sum_{j \in J} c_i(\vec{v}_i) \cdot v_{ij}$. This is clearly ex-post IR because $c_i(\vec{v}_i) \leq 1$. Also, let's compute the expected price bidder $i$ pays when bidding $\vec{v}_i$:

$$\sum_{j} c_i(\vec{v}_i)v_{ij}\pi_{ij}(\vec{v}_i) = c_i(\vec{v}_i)\sum_j v_{ij} \pi_{ij}(\vec{v}_i) = q_i(\vec{v}_i).$$

So we can do this simple transformation to turn an ex-interim IR mechanism into an ex-post IR mechanism without any loss in revenue. One should observe that this transformation may cause bidders to sometimes pay more than their budget (even though the budget constraint is still respected in expectation after our transformation). Unfortunately, this problem is unavoidable in the following sense: the optimal ex-post IR mechanism that respects budget constraints may make strictly less revenue than the optimal ex-interim IR mechanism that respects budget constraints. Here is a simple example that illustrates this on a single item and two bidders. Each bidder always values the item at $10$, but has a budget of $5$. Then the optimal ex-interim IR mechanism that respects budgets is to give the item to each player with probability $1/2$ and charge them $5$ no matter who gets the item. The optimal ex-post IR mechanism that respects budgets is to give the item to each player with probability $1/2$ and charge the winner $5$.

%
%

\section{Proof of Theorem~\ref{thm:symmetry}}\label{app:symmetries}

We provide a proof of our symmetry theorem (Theorem~\ref{thm:symmetry}). The outline of our approach is this:
\begin{enumerate}
\item We show that for an arbitrary $\sigma$, if $M$ is BIC and $\mathcal{D}$ has symmetry $\sigma$, $\sigma(M)$ is BIC and $R^M(\mathcal{D}) = R^{\sigma(M)}(\mathcal{D})$.
\item We extend the result to distributions of permutations. That is, let $\mathcal{G}$ be any distribution over permutations in $S_n \times S_m$ that only samples $\sigma$ such that $\mathcal{D}$ has symmetry $\sigma$. Then define the mechanism $\mathcal{G}(M)$ to first sample a $\sigma$ from $\mathcal{G}$, then use $\sigma(M)$. Then if $M$ is BIC, then $\mathcal{G}(M)$ is BIC and $R^{\mathcal{G}(M)}(\mathcal{D}) = R^M(\mathcal{D})$.
\item We show that if $\mathcal{G}$ uniformly samples a permutation from a {\em subgroup} $S \subseteq S_m \times S_n$, then $\mathcal{G}(M)$ respects every permutation in $S$, for all subgroups $S$ and mechanisms $M$.
\item We put everything together and observe that if $\mathcal{G}$ uniformly samples from the subgroup of symmetries that $\mathcal{D}$ has, then for any mechanism $M$, we can create $\mathcal{G}(M)$ that has the same expected revenue as $M$ and respects every symmetry that $\mathcal{D}$ has.
\end{enumerate}
Next we prove the above steps  one-by-one. We prove more general statements catering also for IC, $\epsilon$-IC and $\epsilon$-BIC mechanisms.
\begin{lemma}\label{lem:truthful} If $M$ is an arbitrary IC (BIC,$\epsilon$-IC, $\epsilon$-BIC) mechanism, then for any $\sigma \in S_m \times S_n$:
\begin{enumerate}
\item $\sigma(M)$ is an IC (BIC,$\epsilon$-IC,$\epsilon$-BIC) mechanism; and
\item $R^{\sigma(M)}(\mathcal{D}) = \sum_{\vec{v} \in \text{supp}(\mathcal{D})} R^M(\sigma^{-1}(\vec{v}))Pr[\vec{v} \leftarrow \mathcal{D}]$.
\end{enumerate}

Furthermore, if $\mathcal{D}$ has symmetry $\sigma$, then $R^M(\mathcal{D}) = R^{\sigma(M)}(\mathcal{D})$.
\end{lemma}

\begin{prevproof}{Lemma}{lem:truthful}
It is clear that part $2$ is true given part $1$. If all bidders play truthfully, then on bidder profile $\vec{v}$, $\sigma(M)$ makes revenue equal to exactly $R^M(\sigma^{-1}(\vec{v}))$. Therefore the sum exactly computes the expected revenue. The last part of the lemma is also clear given part $2$. If $\mathcal{D}$ has symmetry $\sigma$, then the sum exactly computes $R^M(\mathcal{D})$.

Now we prove part $1$. We do this by explicitly examining the value of bidder $i$ whose true type is $\vec{v}_i$ for reporting any other $\vec{w}_i$ when the rest of the bids are fixed. Let $\vec{v}$ denote the profile of bids when everything besides bidder $i$ is fixed and he reports $\vec{v}_i$, and let $\vec{w}$ denote the profile of bids when everything besides bidder $i$ is the same as $\vec{v}$, but bidder $i$ reports $\vec{w}_i$ instead. By the definition of $\sigma(M)$ we have the following equation. (Recall that $U(\vec{v}_i,M_i(\vec{w}))$ denotes the utility of a bidder with type $\vec{v}_i$ for the expected outcome $M_i(\vec{w})$. We also slightly abuse notation with $\sigma$; when we write $\sigma(i)$ we mean the restriction of the permutation $\sigma$ to $[m]$, etc.)

$$U(\vec{v}_i,[\sigma(M)]_i(\vec{x})) = U(\sigma^{-1}(\vec{v}_i),M_{\sigma^{-1}(i)}(\sigma^{-1}(\vec{x}))).$$
{This holds because $\sigma(M)$ on input $\vec{x}$ offers bidder $i$ the permuted by $\sigma$ lottery offered to bidder $\sigma^{-1}(i)$ by $M$ on bid vector $\sigma^{-1}(\vec{x})$ and charges him the price charged to bidder $\sigma^{-1}(i)$ by $M$ on input $\sigma^{-1}(\vec{x})$. }

Now, because $M$ is an IC mechanism, we know that:
$$U(\sigma^{-1}(\vec{v}_i),M_{\sigma^{-1}(i)}(\sigma^{-1}(\vec{v}))) \geq U(\sigma^{-1}(\vec{v}_i),M_{\sigma^{-1}(i)}(\sigma^{-1}(\vec{w}))).$$

And combining this inequality with the above equality, we get exactly that:
$$U(\vec{v}_i,[\sigma(M)]_i(\vec{v})) \geq U(\vec{v}_i,[\sigma(M)]_i(\vec{w})).$$

Since the above was shown for all $i$ and for all $\vec{v} = (\vec{v}_i~;~\vec{v}_{-i})$ and $\vec{w} = (\vec{w}_i~;~\vec{v}_{-i})$, we get that $\sigma(M)$ is an IC mechanism. If instead of being IC $M$ were BIC, $\epsilon$-IC, or $\epsilon$-BIC, the incentive guarantee we have for $M$ still falls exactly through for $\sigma(M)$.
\end{prevproof}

\begin{corollary}\label{cor:important} Let $\mathcal{G}$ denote any distribution over elements of $S_m \times S_n$. For an IC (BIC,$\epsilon$-IC,$\epsilon$-BIC) mechanism $M$, let $\mathcal{G}(M)$ denote the mechanism that samples an element $\sigma$ from $\mathcal{G}$, and then uses the mechanism $\sigma(M)$. Then for all $\cal G$:

\begin{enumerate}
\item $\mathcal{G}(M)$ is an IC (BIC,$\epsilon$-IC,$\epsilon$-BIC) mechanism; and
\item $R^{\mathcal{G}(M)}(\mathcal{D}) = \sum_{\sigma} R^{\sigma(M)}(\mathcal{D}) Pr[\sigma \leftarrow \mathcal{G}]$.
\end{enumerate}

Furthermore, if $\mathcal{G}$ samples only $\sigma$ such that $\mathcal{D}$ has symmetry $\sigma$, then $R^M(\mathcal{D}) = R^{\mathcal{G}(M)}(\mathcal{D})$.
\end{corollary}

\begin{prevproof}{Corollary}{cor:important}
It is clear that, because each $\sigma(M)$ is an IC mechanism, randomly sampling an IC mechanism will result into an IC mechanism. The second claim is also clear by linearity of expectation. The final claim is clear because $R^{\sigma(M)}(\mathcal{D}) = R^M(\mathcal{D})$, if ${\cal D}$ has symmetry $\sigma$, so taking a weighted average of them will still yield $R^M(\mathcal{D})$. We can replace IC by BIC, $\epsilon$-IC, or $\epsilon$-BIC in our argument.
\end{prevproof}

\begin{lemma} \label{lem:final lemma symm}Let $\mathcal{G}$ sample a permutation uniformly at random from a subgroup $S$ of $S_m \times S_n$. Then $\mathcal{G}(M)$ respects every permutation in $S$.
\end{lemma}

\begin{prevproof}{Lemma}{lem:final lemma symm} For any $\vec{v}$, the outcome $\mathcal{G}(M)(\vec{v})$ is:

$$\mathcal{G}(M)(\vec{v}) = \sum_{\sigma \in S} \frac{\sigma(M(\sigma^{-1}(\vec{v})))}{|S|}.$$

 Because $S$ is a subgroup, for any $\tau \in S$, we can write:

$$\mathcal{G}(M)(\tau(\vec{v})) = \sum_{\sigma \in S} \frac{\tau\sigma(M( (\tau\sigma)^{-1}(\tau(\vec{v}))))}{|S|} = \sum_{\sigma \in S} \frac{\tau\sigma(M( \sigma^{-1}(\vec{v})))}{|S|} = \tau(\mathcal{G}(M)(\vec{v})).$$
This is exactly the statement that $\mathcal{G}(M)$ respects symmetry $\tau$. So all $\tau \in S$ are respected by $\mathcal{G}(M)$.
\end{prevproof}

The final step in the proof of Theorem~\ref{thm:symmetries} is just observing that if $S$ denotes the set of symmetries of $\mathcal{D}$, then $S$ is in fact a subgroup because the definition of symmetry immediately yields that if $\mathcal{D}$ has symmetries $\sigma$ and $\tau$, it also has symmetries $\sigma^{-1}$ and $\sigma \tau$.

\section{Proofs Omitted from Section \ref{sec:LP}}\label{app:proof LP}

\begin{prevproof}{Lemma}{lem:simple reduction works for k-items}
Suppose that $S$ contains all bidder permutations and that every marginal of ${\cal D}$ has size at most $c$. We claim that the set of representative value vectors has size {$|E| \leq (m+1)^{c^n}$}. Indeed, each equivalence class is uniquely determined by the number of bidders of each type. There are $c^n$ types of bidders, and {$0$ up to at most $m$} bidders per type, so {$(m+1)^{c^n}$} equivalence classes of value vectors in total. Hence, the LP of Figure~\ref{fig:succinct k-items LP} has $O(n(m+1)^{c^n})$ variables and constraints. If both $c$ and $n$ are constants, the size of the LP is polynomial.
\end{prevproof}

\begin{prevproof}{Theorem}{thm:monotone}
Assume for a contradiction that $M$ is item-symmetric and BIC but not strongly monotone. Let then $\vec{v}_i^*$ be a bidder profile that breaks strong-monotonicity, i.e. $\pi_{ij}(\vec{v}_i^*) > \pi_{ij'}(\vec{v}_i^*)$ and $v_{ij}^* < v_{ij'}^*$ for a pair of items $j$, $j'$. We show that bidder $i$ of type $\vec{v}_i^*$ can strictly increase her expected utility by swapping her values for items $j$ and $j'$. Indeed, because $\mathcal{D}$ is item-symmetric, we know that the distribution of complete bidder profiles satisfies the following for any item permutation $\sigma$:
$$Pr[ \vec{w}|\vec{w}_i = \vec{v}^*_i] = Pr[\sigma(\vec{w})|\vec{w}_i = \sigma(\vec{v}^*_i)].$$
Now because $M$ is item-symmetric, we also know that $M(\sigma(\vec{w})) = \sigma(M(\vec{w}))$, for all $\vec{w}$. Putting these together, we see that we must have $\pi_{ij}(\vec{v}_i^*) = \pi_{i\sigma(j)}(\sigma(\vec{v}_i^*))$, for all item permutations $\sigma$. Letting $\sigma$ be the permutation that swaps items $j$ and $j'$ shows that when bidder $i$ of type $\vec{v}_i^*$ swaps her values for $j$ and $j'$, she simply switches $\pi_{ij}(\vec{v}_i^*)$ and $\pi_{ij'}(\vec{v}_i^*)$, strictly increasing her utility.

For the second part of the theorem, suppose  that $M$ is any item-symmetric $\epsilon$-BIC mechanism that is not strongly monotone. Again let $\vec{v}^*_i$ be a bidder profile that breaks strong-monotonicity with $\pi_{ij}(\vec{v}^*_i) > \pi_{ij'}(\vec{v}^*_i)$ and $v_{ij}^* < v_{ij'}^*$ for a pair of items $j$ and $j'$. Let then $M'$ be the mechanism that does the following.  If bidder $i$ reports $\vec{w}_i = \tau(\vec{v}^*_i)$, for some $\tau \in S_n$, then pick a random such $\tau$, swap the bidder's values for items $\tau(j)$ and $\tau(j')$, and run $M$. Otherwise just run $M$. It is clear that, $M'$ has the exact same expected revenue as $M$, when played truthfully, because $\mathcal{D}$ is item-symmetric. Observe also that we have added no new alternatives for dishonest bidders to consider reporting, maintained the item symmetry of the mechanism, and made some bidders strictly happier for the outcomes that the mechanism gives them. Therefore, $M'$ is still item-symmetric and $\epsilon$-BIC and makes the same revenue as $M$, when played truthfully. Additionally we have corrected one violation of strong-monotonicity. Iterating this process for a finite number of steps will yield an $\epsilon$-BIC item-symmetric strongly monotone mechanism with the same expected revenue as $M$.
\end{prevproof}

\begin{prevproof}{Lemma}{lem: succinct LP for k-bidders}
Recall the definition of the set $E$ of representative bidder profiles from Section~\ref{subsec:bidder symmetries}/Appendix~\ref{app:succinct LPs}. In addition to this set we define, for each bidder $i$, a set of representative types, $E_i$, for this bidder. $E_i$ contains only value vectors $\vec{v}_i$ satisfying $v_{i1} \geq v_{i2} \geq \ldots \geq v_{in}$. By Observation \ref{obs:monotone}, if {a mechanism is item-symmetric} and no type in $E_i$ wishes to misreport another type in $E_i$, for all $i$, then the mechanism is BIC. Our resulting LP is shown in Figure~\ref{fig:succinct k-bidders LP}. In total, we have {$O(n|E|\sum_i|E_i|)$} variables and {$O(n |E| \sum_i |E_i|^2)$} constraints. Given that the value distribution is symmetric under every item permutation, we have that {$|E| \leq (n+1)^{c^m}$}. Indeed, there are $c^m$ possible ways the $m$ bidders like an item, and the question is how many items are liked in each of the possible ways.  We also know that {$|E_i| \leq (n+1)^c$}, again because choosing the number of items valued by a bidder at each of the $c$ possible values uniquely determines an element of $E_i$. It follows that when $c$ and $m$ are constants the size of the LP is polynomial.
\end{prevproof}

\section{Succinct LP Formulations} \label{app:succinct LPs}

Figures~\ref{fig:succinct k-items LP} and~\ref{fig:succinct k-bidders LP} show the succinct LPs that can be used to compute the optimal mechanisms for the BIC $k$-items and the BIC $k$-bidders problems respectively. Details for satisfying the supply and demand constraints with probability $1$, and ex-post IR modifications are exactly the same as in Appendices~\ref{app:feasible} and~\ref{app:ex-post IR}. In both figures, $E$ denotes a set of representatives under the equivalence relation defined by bidder or item symmetries. In Figure \ref{fig:succinct k-bidders LP}, $E_i$ denotes the set of types for bidder $i$ such that $v_{i1} \geq \ldots \geq v_{in}$.
\begin{figure}[h!]
\colorbox{MyGray}{
\begin{minipage}{\textwidth}
{\noindent\textbf{Variables:}
\begin{itemize}
\item $\phi_{ij}(\vec{v})$, for all bidders $i$, items $j$, and bidder profiles $\vec{v} \in E$ ($mn|E|$).
\item $\pi_{ij}(\vec{v}_i)$, for all bidders $i$, items $j$, and possible types for bidder $i$, $\vec{v}_i$ ($mnc^n$).
\item $p_i(\vec{v})$, for all bidders $i$, and bidder profiles $\vec{v} \in E$ ($m|E|$).
\item $q_i(\vec{v}_i)$, for all bidders $i$, possible types for bidder $i$, $\vec{v}_i$ ($mc^n$).
\end{itemize}
\textbf{Constraints:}
\begin{itemize}
\item (Precomputed Weights): ${\rm aux}(i',j',i,j, \vec{w}, \vec{v}_i) = \sum_{\sigma \in S: \sigma(\vec{w})_i = \vec{v}_i \wedge \sigma^{-1}(i,j)=(i',j')} \Pr[\sigma(\vec{w})~|~\sigma(\vec{w})_i =\vec{v}_i]$, for all bidders $i, i'$, items $j, j'$, bidder profiles $\vec{w} \in E$, and types $\vec{v}_i$.~\footnote{{For the BIC $k$-items problem, these weights can be computed efficiently. If $j \neq j'$, then ${\rm aux}(i',j',i,j,\vec{w},\vec{v}_i) = 0$, as there is no $\sigma\in S$ such that $\sigma^{-1}(j)\neq j$. If $\vec{w}_{i'} \neq \vec{v}_i$, then ${\rm aux}(i',j',i,j,\vec{w},\vec{v}_i) = 0$. Otherwise, ${\rm aux}(i',j',i,j,\vec{w},\vec{v}_i) = \frac{(m-1)! \cdot Pr[\vec{w}]}{Pr[\vec{x}_{i'} \leftarrow \mathcal{D}_{i'},\vec{x}_{i'} = \vec{v}_i]}$.}}
\item $\pi_{ij}(\vec{v}_i)= \sum_{\vec{w} \in E} \sum_{i',j'}  \phi_{i',j'}(\vec{w}) {\rm aux}(i',j',i,j,\vec{w}, \vec{v}_i)$,~\footnote{Justification:
\begin{align*}
\pi_{ij}(\vec{v}_i) &= \sum_{\vec{w} \in E} \sum_{\sigma \in S: \sigma(\vec{w})_i = \vec{v}_i} Pr[\sigma(\vec{w})~|~ \sigma(\vec{w})_i = \vec{v}_i]\phi_{\sigma^{-1}(i,j)}(\vec{w})\\
&= \sum_{\vec{w} \in E} \sum_{i',j'}  \phi_{i',j'}(\vec{w}) \sum_{\sigma \in S: \sigma(\vec{w})_i = \vec{v}_i~\wedge~\sigma^{-1}(i,j)=(i',j')} Pr[\sigma(\vec{w})~|~ \sigma(\vec{w})_i = \vec{v}_i] \\
&= \sum_{\vec{w} \in E} \sum_{i',j'}  \phi_{i',j'}(\vec{w}) \cdot {\rm aux}(i',j',i,j,\vec{w}, \vec{v}_i).
\end{align*}
}, for all $i,j,\vec{v}_i$ ($mn c^n$).
\item $q_i(\vec{v}_i)= \sum_{\vec{w} \in E} \sum_{i',j'}  p_{i'}(\vec{w}) {\rm aux}(i',j',i,1,\vec{w}, \vec{v}_i)$, for all $i, \vec{v}_i \in E_i$ ($m c^n$).
\item $0 \leq \phi_{ij}(\vec{v}) \leq 1$, for all $i,j,\vec{v} \in E$ ($mn|E|$).
\item $\sum_i \phi_{ij}(\vec{v}) \leq 1$, for all $j,\vec{v} \in E$ ($n|E|$).
\item $\sum_j \phi_{ij}(\vec{v}) \leq C_i$, for all $i,\vec{v} \in E$ ($m|E|$).
\item $p_i(\vec{v}) \leq B_i$ for all $i,\vec{v} \in E$ ($m|E|$).
\item $\sum_{j} v_{ij}\pi_{ij}(\vec{v}_i) - q_i(\vec{v}_i) \geq 0$, for all $i,\vec{v}_i$ ($mc^n$).
\item $\sum_j v_{ij}\pi_{ij}(\vec{v}_i) - q_i(\vec{v}_i) \geq \sum_j v_{ij}\pi_{ij}(\vec{v}'_i) - q_i(\vec{v}'_i)$, for all $i,\vec{v}_i,\vec{v}'_i$ ($mc^{2n}$).
\end{itemize}
\textbf{Maximizing:}\\
$\sum_{i,\vec{v} \in E} p_i(\vec{v})Pr[\cup_{\sigma \in S} \sigma(\vec{v}))]$.}
\end{minipage}
}\caption{{Succinct BIC k-items LP. In parentheses at the end of each type of variable/constraint is an upper bound on the number of such variables/constraints.}}\label{fig:succinct k-items LP}
\end{figure}

\begin{figure}[h!]
\colorbox{MyGray}{
\begin{minipage}{\textwidth} {
\noindent\textbf{Variables:}
\begin{itemize}
\item $\phi_{ij}(\vec{v})$, for all bidders $i$, items $j$, and bidder profiles $\vec{v} \in E$ ($mn|E|$).
\item $\pi_{ij}(\vec{v}_i)$, for all bidders $i$, items $j$, and $\vec{v}_i \in E_i$ {($n\sum_i |E_i|$).}
\item $p_i(\vec{v})$, for all bidders $i$, and bidder profiles $\vec{v} \in E$ ($m|E|$).
\item $q_i(\vec{v}_i)$, for all bidders $i$, $\vec{v}_i \in E_i$ {($\sum_i |E_i|$).}
\end{itemize}
\textbf{Constraints:}
\begin{itemize}
\item (Precomputed Weights): ${\rm aux}(i',j',i,j, \vec{w}, \vec{v}_i) = \sum_{\sigma \in S: \sigma(\vec{w})_i = \vec{v}_i \wedge \sigma(i,j)^{-1}=(i',j')} \Pr[\sigma(\vec{w})~|~\sigma(\vec{w})_i =\vec{v}_i]$, for all bidders $i, i'$, items $j, j'$, bidder profiles $\vec{w} \in E$, and types $\vec{v}_i$.~\footnote{{For the BIC $k$-bidders problem, these weights can be computed efficiently. If $i \neq i'$, then ${\rm aux}(i',j',i,j,\vec{w},\vec{v}_i) = 0$. If $\vec{w}_i \not \sim_S \vec{v}_i$, then ${\rm aux}(i',j',i,j,\vec{w},\vec{v}_i) = 0$. If $\vec{w}_{ij'} \neq \vec{v}_{ij}$, then ${\rm aux}(i',j',i,j,\vec{w},\vec{v}_i) = 0$. Otherwise, let the $c$ possible values for the items be $u_1,\ldots,u_c$. Let $n_k$ denote the number of items $j$ with $\vec{v}_{ij} = u_k$. Then there are $\prod_k n_k!$ different permutations $\sigma$ such that $\sigma(\vec{w})_i = \vec{v}_i$. If $v_{ij} = u_a$, then the number of permutations such that $\sigma(\vec{w})_i = \vec{v}_i$ and $\sigma(i,j') = (i,j)$ is $(\prod_k n_k!)/n_a$. Then ${\rm aux}(i',j',i,j,\vec{w},\vec{v}_i) = \frac{\prod_k n_k! Pr[\vec{w}]}{n_aPr[\vec{x}_i \leftarrow \mathcal{D}_i, \vec{x}_i = \vec{v}_i]}$}.}
\item $\pi_{ij}(\vec{v}_i)= \sum_{\vec{w} \in E} \sum_{i',j'}  \phi_{i',j'}(\vec{w}) \cdot {\rm aux}(i',j',i,j,\vec{w}, \vec{v}_i)$, for all $i,j,\vec{v}_i$ ($mn c^n$)
\item $q_i(\vec{v}_i)= \sum_{\vec{w} \in E} \sum_{i',j'}  p_{i'}(\vec{w}) \cdot {\rm aux}(i',j',i,1,\vec{w}, \vec{v}_i)$, for all $i, \vec{v}_i \in E_i$ ($\sum_i |E_i|$).
\item $0 \leq \phi_{ij}(\vec{v}) \leq 1$, for all $i,j,\vec{v} \in E$ ($mn|E|$).
\item $\sum_i \phi_{ij}(\vec{v}) \leq 1$, for all $j,\vec{v} \in E$ ($n|E|$).
\item $\sum_j \phi_{ij}(\vec{v}) \leq C_i$, for all $i,\vec{v} \in E$ ($m|E|$).
\item $p_i(\vec{v}) \leq B_i$, for all $i,\vec{v} \in E$ ($m|E|$).
\item $\sum_{j} v_{ij}\pi_{ij}(\vec{v}_i) - q_i(\vec{v}_i) \geq 0$, for all $i,\vec{v}_i \in E_i$ {($\sum_i |E_i|$).}
\item $\sum_j v_{ij}\pi_{ij}(\vec{v}_i) - q_i(\vec{v}_i) \geq \sum_j v_{ij}\pi_{ij}(\vec{v}'_i) - q_i(\vec{v}'_i)$, for all $i,\vec{v}_i \in E_i,\vec{v}'_i \in E_i$ {($\sum_i |E_i|^2$).}
\item $\pi_{ij}(\vec{v}_i) \geq \pi_{i(j+1)}(\vec{v}_i)$, for all $i,j,\vec{v}_i \in E_i$ {($n \sum_i |E_i|$).}
\end{itemize}
\textbf{Maximizing:}\\
$\sum_{i,\vec{v} \in E} p_i(\vec{v})Pr[\cup_{\sigma \in S} \sigma(\vec{v}))]$.}
\end{minipage}} \caption{Succinct BIC $k$-bidders LP. In parentheses at the end of each type of variable/constraint is an upper bound on the number of such variables/constraints. }\label{fig:succinct k-bidders LP}
\end{figure}

\section{Omitted Proofs from Section~\ref{sec:epsilon-truthful}}
\label{app:technical} \label{app:bi-criterion}

In this section we relate the optimal revenue achievable under ``similar'' value distributions, and provide the proof of Lemma~\ref{lem:deltaIC}. {Throughout the section, we let $T$ denote the maximum number of items that can be awarded by a feasible mechanism. For $k$-items this is $k$, for $k$-bidders this is $\min\{n,\sum_i C_i\}$.}


\begin{lemma}\label{lem:add by delta} Suppose that $\mathcal{D}$ and $\mathcal{D'}$ can be coupled so that, with probability $1$, whenever $\vec{v} \leftarrow \mathcal{D}$ and $\vec{v}'\leftarrow \mathcal{D'}$ are jointly sampled under the coupling, it holds that $v'_{ij} = v_{ij} + \delta$, for all $i,j$. Then for all $\epsilon$, $R_\epsilon^{OPT}(\mathcal{D}) \geq R_\epsilon^{OPT}(\mathcal{D'}) - {\delta T}$.
\end{lemma}

\begin{prevproof}{Lemma}{lem:add by delta} Throughout this proof, $\ind$ represents the all-ones vector. Let $M'$ be any $\epsilon$-BIC mechanism for $\mathcal{D'}$. We define a new mechanism $M$ for ${\cal D}$ such that $M(\vec{v}) = M'(\vec{v}+\delta \cdot \ind)$, except that we lower the price paid by bidder $i$ by $\delta \cdot \sum_j \phi'_{ij}(\vec{v}+\delta \cdot \ind)$, i.e. $\delta$ times the expected number of items given to bidder $i$ by $M'(\vec{v}+\delta \cdot \ind)$). Then, for all $i$, $\vec{v}_i, \vec{w}_i, \vec{v}_{-i}$ and corresponding $\vec{v}'_i := \vec{v}_i +\delta \cdot \ind, \vec{w}'_i:=\vec{w}_i +\delta \cdot \ind, \vec{v}'_{-i}:=\vec{v}_{-i}+\delta \cdot \ind$, we have:

$$U(\vec{v_i},M_i(\vec{w}_i~;~\vec{v}_{-i})) = U(\vec{v_i}',M'_i(\vec{w}'_i~;~\vec{v}'_{-i})).$$
Hence because $M'$ is $\epsilon$-BIC under ${\cal D}'$, $M$ is $\epsilon$-BIC under ${\cal D}$. It is also clear that the difference in expected revenue of the two mechanisms under the two distributions is exactly $\delta$ times the expected number of items given out by $M'$, which is at most {$T$}.
\end{prevproof}

\begin{lemma}\label{lem:coupling} Suppose that $\mathcal{D} = \times_i {\cal D}_i$ and $\mathcal{D'}= \times_i {\cal D}'_i$ are product distributions over bidders and suppose that, for all $i$, there is a coupling of ${\cal D}_i$ and ${\cal D}'_i$ so that, with probability $1$, if $\vec{v}_i \leftarrow \mathcal{D}_i$ and $\vec{v}'_i\leftarrow \mathcal{D}_i'$ are jointly sampled under this coupling, it holds that $v_{ij} \leq v'_{ij} \leq v_{ij} + \delta$, for all $j$. Then $R_{{\delta} + \epsilon}^{OPT}(\mathcal{D'}) \geq R_{\epsilon}^{OPT}(\mathcal{D})$, for all $\epsilon$.
\end{lemma}

\begin{prevproof}{Lemma}{lem:coupling}
For all $i$, the coupling whose existence is certified in the statement of the lemma, implies the existence of a (possibly randomized) mapping $f_i^R$ such that the distribution that samples $\vec{v}_i'$ from ${\cal D}_i'$ and outputs the pair $(\vec{v}_i,\vec{v}_i')$, where $\vec{v}_i$ is a random vector sampled from $f_i^R(\vec{v}_i')$, is a valid coupling of ${\cal D}_i$ and ${\cal D}'_i$ satisfying $v_{ij} \le v_{ij}' \le v_{ij}+\delta$, for all $j$, with probability $1$. Let then $f^R$ be the random mapping which on input $\vec{v}'$ samples, for all $i$, a random $\vec{v}_i$ from $f^R_i(\vec{v}_i')$ and outputs $(\vec{v}_1,\ldots,\vec{v}_m)$.

Now consider any $\epsilon$-BIC mechanism $M$ for $\mathcal{D}$, and define the mechanism $M'$ for ${\cal D}'$, which on input $\vec{v}'$ samples a random $\vec{v}$ from $f_R(\vec{v}')$ and outputs $M(\vec{v})$. It is obvious that $R^{M'}(\mathcal{D'}) = R^M(\mathcal{D})$. To conclude the proof of the lemma it suffices to show that $M'$ is $(\epsilon+{\delta})$-BIC for bidders sampled from $\mathcal{D'}$.  Indeed, we have from the fact that $M$ is $\epsilon$-BIC that for all $i$, $\vec{v}_i$ and $\vec{w}_i$:

$$\mathbb{E}_{\vec{v}_{-i} \sim {\cal D}_{-i}}[U(\vec{v_i},M_i(\vec{v}))] \geq \mathbb{E}_{\vec{v}_{-i} \sim {\cal D}_{-i}} [U(\vec{v_i},M_i(\vec{w}_i~;~\vec{v}_{-i}))] - \epsilon {\cdot \sum_j \pi_{ij}(\vec{w}_i)}.$$

Now fix $i$, $\vec{v}_i'$ and $\vec{w}_i'$. We have:
\begin{align}
\mathbb{E}_{\vec{v}'_{-i} \sim {\cal D}'_{-i}} [U(\vec{v_i}',M_i'(\vec{w}_i'~;~\vec{v}'_{-i}))] &= \mathbb{E}_{\vec{v}'_{-i} \sim {\cal D}'_{-i}} [U(\vec{v_i}', \mathbb{E}_{\vec{w}_i \sim f^R_i(\vec{w}_i'), \vec{v}_{-i} \sim f^R_{-i}(\vec{v}'_{-i})} M_i(\vec{w}_i~;~\vec{v}_{-i}))]\notag\\
&= \mathbb{E}_{\vec{v}'_{-i} \sim {\cal D}'_{-i}, \vec{w}_i \sim f^R_i(\vec{w}_i'), \vec{v}_{-i} \sim f^R_{-i}(\vec{v}'_{-i})} [U(\vec{v_i}', M_i(\vec{w}_i~;~\vec{v}_{-i}))] \notag\\
&= \mathbb{E}_{\vec{v}_{-i} \sim {\cal D}_{-i}, \vec{w}_i \sim f^R_i(\vec{w}_i')} [U(\vec{v_i}', M_i(\vec{w}_i~;~\vec{v}_{-i}))] \label{eq:tiresome 1}
\end{align}
Using Eq.~\eqref{eq:tiresome 1} we have:
\begin{align*}\mathbb{E}_{\vec{v}'_{-i} \sim {\cal D}'_{-i}} [U(\vec{v_i}',M_i'(\vec{v}_i'~;~\vec{v}'_{-i}))] &=\mathbb{E}_{\vec{v}_{-i} \sim {\cal D}_{-i}, \vec{v}_i \sim f^R_i(\vec{v}_i')} [U(\vec{v_i}', M_i(\vec{v}_i~;~\vec{v}_{-i}))]\\ &\ge \mathbb{E}_{\vec{v}_{-i} \sim {\cal D}_{-i}, \vec{v}_i \sim f^R_i(\vec{v}_i')} [U(\vec{v_i}, M_i(\vec{v}_i~;~\vec{v}_{-i}))]\\
&= \mathbb{E}_{\vec{v}_i \sim f^R_i(\vec{v}_i')} \mathbb{E}_{\vec{v}_{-i} \sim {\cal D}_{-i}} [U(\vec{v_i}, M_i(\vec{v}_i~;~\vec{v}_{-i}))]\\
&\ge \mathbb{E}_{\vec{w}_i \sim f^R_i(\vec{w}_i'), \vec{v}_i \sim f^R_i(\vec{v}_i')} \mathbb{E}_{\vec{v}_{-i} \sim {\cal D}_{-i}} [U(\vec{v_i}, M_i(\vec{w}_i~;~\vec{v}_{-i})) - \epsilon {\cdot \sum_j \pi_{ij}(\vec{w}_i)}]\\
&{= \mathbb{E}_{\vec{w}_i \sim f^R_i(\vec{w}_i'), \vec{v}_i \sim f^R_i(\vec{v}_i')} \mathbb{E}_{\vec{v}_{-i} \sim {\cal D}_{-i}} [U(\vec{v_i}, M_i(\vec{w}_i~;~\vec{v}_{-i}))] - \epsilon \cdot \sum_j \pi'_{ij}(\vec{w}'_i)},
\end{align*}
where for the last inequality we used that $M$ is $\epsilon$-BIC. Similarly,
\begin{align*}
\mathbb{E}_{\vec{v}'_{-i} \sim {\cal D}'_{-i}} [U(\vec{v_i}',M_i'(\vec{w}_i'~;~\vec{v}'_{-i}))] &=\mathbb{E}_{\vec{v}_{-i} \sim {\cal D}_{-i}, \vec{w}_i \sim f^R_i(\vec{w}_i')} [U(\vec{v_i}', M_i(\vec{w}_i~;~\vec{v}_{-i}))]\\ 
&\le \mathbb{E}_{\vec{v}_{-i} \sim {\cal D}_{-i}, \vec{w}_i \sim f^R_i(\vec{w}_i'), \vec{v}_i \sim f^R_i(\vec{v}_i')} [U(\vec{v_i}+\delta \ind, M_i(\vec{w}_i~;~\vec{v}_{-i}))] \\ 
&= \mathbb{E}_{\vec{v}_{-i} \sim {\cal D}_{-i}, \vec{w}_i \sim f^R_i(\vec{w}_i'), \vec{v}_i \sim f^R_i(\vec{v}_i')} [U(\vec{v_i}, M_i(\vec{w}_i~;~\vec{v}_{-i})){  + \delta \cdot \sum_j \pi_{ij}(\vec{w}_i)]}\\
&= \mathbb{E}_{\vec{w}_i \sim f^R_i(\vec{w}_i'), \vec{v}_i \sim f^R_i(\vec{v}_i')}  \mathbb{E}_{\vec{v}_{-i} \sim {\cal D}_{-i}} [U(\vec{v_i}, M_i(\vec{w}_i~;~\vec{v}_{-i}))]  + {\delta \cdot \sum_j \pi'_{ij}(\vec{w}'_i)}.
\end{align*}
Combining the above it follows that $M'$ is $(\epsilon+\delta)$-BIC.
\end{prevproof}


\begin{prevproof}{Lemma}{lem:deltaIC}
Let $\mathcal{D}$ denote the original distribution. Let $\mathcal{D}''$ denote the distribution that first samples from $\mathcal{D}$, then rounds every value up to the nearest multiple of $\delta$ (if the sampled value from ${\cal D}$ is exactly at an integer multiple of $\delta$ it is rounded up to the next integer multiple of $\delta$). Let $\mathcal{D}'$ denote the distribution that first samples from $\mathcal{D}$, then rounds every value down to the nearest multiple of $\delta$ (same definition as in the statement of Lemma \ref{lem:deltaIC}). Then it is clear that $\mathcal{D}'$ and $\mathcal{D''}$ satisfy the hypotheses of Lemma \ref{lem:add by delta}, so we have:

$$R^{OPT}_{{\delta}}(\mathcal{D}')\geq R^{OPT}_{{\delta}}(\mathcal{D''}){- {\delta T}}.$$

It is also clear that $\mathcal{D}$ and $\mathcal{D}''$ satisfy the hypotheses of Lemma \ref{lem:coupling}, so we have:

$$R^{OPT}_{{\delta}}(\mathcal{D}'') \geq {R^{OPT}(\mathcal{D})}.$$

Putting both together, we get that:

$$R^{OPT}_{{\delta}}(\mathcal{D}') \geq R^{OPT}(\mathcal{D}) - {\delta T}.$$

It follows immediately from the fact that $M'$ is {$\delta$}-BIC for consumers in $\mathcal{D}'$ that $M$ is {$2\delta$}-BIC for consumers in $\mathcal{D}$  {(via the  argument given in the proof of Lemma~\ref{lem:coupling}) and it is obvious that $R^{M'}(\mathcal{D'}) = R^M(\mathcal{D})$.} This completes the proof of the lemma.
\end{prevproof}

\section{Proof of Theorem \ref{thm:epsilon-BIC to BIC}}\label{app:true PTAS}

{Here we prove Theorem \ref{thm:epsilon-BIC to BIC}. We start with a proof outline, and justify each step separately. Unless otherwise stated, every claim applies to both reductions (for the $k$-bidders and the $k$-items problems). {Before starting, we observe that the assumption that $\mathcal{D'}$ is discrete is a \emph{simplifying} assumption and not a \emph{necessary} assumption. {In our proof of Theorem \ref{thm:epsilon-BIC to BIC} below, we point out the key modification that makes it work for continuous ${\cal D}'$ at an additional loss of $O(\frac{\delta}{\eta}T)$ in revenue.

Throughout the proofs, we will use $T_i$ to denote the maximum number of items that are possibly awarded to bidder $i$ {by a feasible mechanism}, and $T$ to denote the maximum number of items that are possibly awarded {by a feasible mechanism}. For $k$-bidders, $T_i = C_i$, {$T= \min\{n,\sum_i C_i\}$}. For $k$-items, {$T_i = \min\{k, C\}$, $T = k$}.  Here is a brief outline of our proof. Let $M_2$ denote the mechanism output by our reduction; that the output of the reduction is a valid mechanism will be justified in what follows.}


\begin{enumerate}
\item If bidder $i$ plays $M_2$ truthfully, then the distribution of surrogates matched to bidder $i$ is $\mathcal{D}'_i$.
\item For $k$-bidders, because $M$, $\mathcal{D}$ and {${\cal D}'$} are item-symmetric, no bidder gains by lying about her ordering.
\item Because each bidder is participating in a VCG auction, and the value of each edge is calculated exactly given that all other bidders tell the truth, $M_2$ is BIC.
\item The revenue we make from bidder $i$ is at least the price paid by their surrogate if bidder $i$ is matched in VCG, and $0$ otherwise.
\item There exists a high cardinality matching with positive edge weights. If VCG used this matching, we would make almost as much revenue as $M$. We show that because of the rebates (Phase One, Step 2), the VCG matching makes almost as much revenue as this matching.
\item Therefore, we make a good approximation to $R^M(\mathcal{D'})$, which is in turn a good approximation to $R^{M_1}(\mathcal{D}')$.
\end{enumerate}

\begin{lemma}[\cite{HKM}]\label{lem:HKM} If all bidders play $M_2$ truthfully, then the distribution of the surrogate chosen for bidder $i$ is exactly $\mathcal{D}'_i$.
\end{lemma}
\begin{proof} Imagine changing the order of sampling in the experiment. First, sample the $r$ surrogates i.i.d. from $\mathcal{D}'_i$, and $r$ replicas i.i.d. from $\mathcal{D}_i$. For the $k$-bidders algorithm only, pick a random ordering $\sigma$ and permute each surrogate and replica to respect that ordering. Finally, pick a replica uniformly at random among the sampled ones and decide that this was in fact bidder $i$. The distribution of surrogates, replicas, and bidder $i$ is exactly the same as that of the algorithm assuming bidder $i$ truthfully reports his type. In addition, we can compute the VCG matching once all the replicas and surrogates have been sampled before deciding which replica is bidder $i$. Because we choose a random replica to be bidder $i$, the surrogate chosen to represent her will also be chosen uniformly at random, regardless of the surrogates' types. So we can see that the process of choosing a surrogate is just sampling $r$ times from $\mathcal{D}'_i$ independently, permuting each sample to respect a random ordering (in the $k$-bidders problem only), and outputting a random surrogate among the sampled (and possibly permuted) ones. In the $k$-items reduction, the output surrogate is clearly distributed according to ${\cal D}'_i$. In the $k$-bidders reduction, because ${\cal D}'_i$ is item-symmetric and we used a random permutation to permute all sampled surrogates, we still output a surrogate distributed according to ${D}'_i$.
\end{proof}

\begin{lemma}\label{lem:ordering} In $M_2$ resulting from the $k$-bidders reduction, no bidder $i$ has an incentive to report any $\vec{w}_i$ such that there is some $j$ and $j'$ for which $w_{ij} > w_{ij'}$ when the true type $\vec{v}_i$ of the bidder satisfies $v_{ij} < v_{ij'}$.
\end{lemma}

\begin{proof} Let $\sigma(\vec{w_i})$ be such that $\sigma(\vec{w_i})_j {\ge} \sigma(\vec{w_i})_{j'}$ if and only if $v_{ij} {\ge} v_{ij'}$. We show that bidder $i$ would be better off reporting $\sigma(\vec{w_i})$ than $\vec{w_i}$. Indeed, we can couple the outcomes of $M_2$ on $\vec{w_i}$ and $\sigma(\vec{w_i})$ in the following way. Whenever a replica is sampled for $i$ to play against, sample the same replica for both experiments. Whenever a surrogate is sampled to represent her, sample the same for both experiments. This set of replicas and surrogates will get permuted to match the ordering of $\vec{w_i}$ and $\sigma(\vec{w_i})$ respectively, so the VCG matching chosen will be exactly the same. So if we let $\vec{s_i}$ denote the surrogate chosen when the bid was $\vec{w_i}$, then $\sigma(\vec{s_i})$ is the surrogate chosen on bid $\sigma(\vec{w_i})$. Because $\sigma$ was chosen so that $\sigma(\vec{s_i})$ is ordered the same as $\vec{v_i}$ and $M$ is strongly monotone, bidder $i$ prefers to be represented by $\sigma(\vec{w_i})$ than by $\vec{w_i}$ when her type is $\vec{v_i}$. So bidder $i$ has no incentive to lie about the relative ordering of $\vec{v_i}$.
\end{proof}

\begin{corollary} $M_2$ is BIC.
\end{corollary}
\begin{proof}
{Fix some bidder $i$ and suppose all other bidders report truthfully. By Lemma \ref{lem:HKM} the surrogates chosen for them are distributed according to $\mathcal{D}'_{-i}$. Hence, when we design the VCG auction for bidder $i$, we correctly compute the expected outcome that will be awarded to each surrogate if that surrogate is chosen for bidder $i$. So in the $k$-items reduction, bidder $i$ faces a VCG mechanism that correctly computes edge-weights; hence the bidder will play the VCG mechanism truthfully. In the $k$-bidders reduction, by Lemma \ref{lem:ordering} bidder $i$ won't misreport her ordering, but could possibly lie about her values (respecting her ordering). Nevertheless, no matter what (ordering respecting) values she reports, she won't affect the distribution of the values of the replicas and surrogates she will compete against in VCG, except for the fact that these are going to be conditioned to respect her ordering. Still, because the edge-weights in the VCG auction are computed correctly, the bidder will report her type truthfully.}
\end{proof}

Now that we know $M_2$ is BIC, we want to compare $R^{M_2}(\mathcal{D})$ to $R^M(\mathcal{D'})$. We observe first that if we are lucky and the VCG matching for every bidder is always perfect, then in fact $R^{M_2}(\mathcal{D}) \geq R^M(\mathcal{D'})$. When bidders are matched in VCG to a surrogate, they pay exactly what their surrogate paid, plus a little extra due to the VCG prices. However, because some edge weights are negative, the VCG matching may not be perfect. So how do we analyze expected revenue in this case?

We look at the expected revenue contributed from a single bidder $i$. Consider again changing the order of sampling in the experiment to sample each replica and surrogate first before deciding which replica is bidder $i$. Then if surrogate $\vec{s_i}$ is matched in VCG, there is a $1/r$ chance that its matched replica will be bidder $i$ and we make expected revenue equal to what $\vec{s_i}$ pays in $M$. If $\vec{s_i}$ is unmatched in VCG, then even if its matched replica is bidder $i$, we make no revenue. This ignores the extra possible revenue from the VCG prices, which we will continue to ignore from now on. So let $p(\vec{s_i})$ denote the expected price that a surrogate with type $\vec{s_i}$ pays in $M$ (over the randomness of the other surrogates), and let $V$ denote the set of surrogates that are matched in VCG. Then the expected revenue of $M_2$ from bidder $i$ is exactly:

$$\sum_{\vec{s_i} \in V} p(\vec{s_i})/r.$$

Recall that the expected revenue of $M$ from bidder $i$ is exactly $\sum_{\forall \vec{s_i}} p(\vec{s_i})/r$. So our goal is to bound the difference between these two sums. But of course, there is also randomness in which surrogates are sampled and what the VCG matching is, so we want to bound the difference in expectation of these two sums. We do this in two steps. First, we show that there exists a matching that is \emph{very} close to perfect in expectation and only uses positive edges. In fact, it is so close to perfect that even if we assume that the unmatched surrogates had the highest possible price we barely lose any revenue in expectation. Unfortunately, this is not necessarily the matching that VCG uses. However, because VCG maximizes social welfare, and we give surrogates a free rebate of {$\eta p_i(\vec{s_i})$}, VCG cannot unmatch too many surrogates that pay a high price. We now quantify these statements and prove them.

Define an equivalence relation on bidders and surrogates where $\vec{v_i} \sim \vec{w_i}$ if when we round both vectors down to the nearest multiple of $\delta$ we get the same vector. Observe that a replica $\vec{v_i}$ and surrogate $\vec{s_i}$ are equivalent only if $v_{ij} \geq s_{ij}$ for all $j$. Therefore, replica $\vec{v_i}$ has positive valuation for the outcome of surrogate $\vec{s_i}$.~\footnote{This is the step where it is helpful to assume that $\mathcal{D}'$ is discrete. If $\mathcal{D}'$ were not discrete, we would not necessarily have $v_{ij} \geq s_{ij}$. However, after giving an additional rebate of $\delta$ for every item received, the conclusion that $\vec{v}_i$ has positive valuation for the surrogate $\vec{s}_i$ holds, which is what matters. The extra rebates result in an extra loss of $O({\delta \over \eta} T)$ of revenue. This loss in revenue comes from: a) actually giving the rebates, which costs at most $\delta T$ in revenue and b) possibly reducing the truthfulness of $M_1$ from $\epsilon$-BIC to at worst $(\epsilon+\delta)$-BIC, which costs an additional $\frac{\delta}{\eta}T$ in revenue (replacing $\epsilon$ by $\epsilon+\delta$ in Lemma~\ref{lem:augmenting}, Corollary~\ref{cor:augmenting}, and the discussion following Lemma~\ref{lem:boundd}).} So a matching that matches replicas only to equivalent surrogates uses only positive edges. {If we let $\beta$ be the number of equivalence classes, then $\beta \le ({1 \over \delta}+1)^k$ (for $k$-items) and $\beta \le (n+1)^{1/\delta+1}$ (for $k$-bidders taking into account that we are only looking at permuted replicas/surrogoates in our transformation). We can use a lemma from \cite{HKM} directly:}

\begin{lemma}[\cite{HKM}] The expected cardinality of a maximal matching that only matches equivalent replicas and surrogates is at least $r - \sqrt{\beta r}$.
\end{lemma}

So if we denote by $X$ the set of surrogates in some maximal matching using only equivalent replicas, then because each $p(\vec{s_i})$ is at most $T_i$ we get:

$$\mathbb{E}\left[\sum_{\vec{s_i} \in X} p(\vec{s_i})/r \right] \geq \mathbb{E} \left [\sum_{\forall \vec{s_i}} p(\vec{s_i})/r \right] - T_i\sqrt{\beta / r}$$

Now we want to bound the expected difference between the expected revenue from the matching $X$ and the VCG matching $V$. We can get from $X$ to $V$ through a disjoint collection of augmenting paths and cycles. We want to show that there cannot be many augmenting paths that unmatch a surrogate with a high $p(\vec{s_i})$.

\begin{lemma} \label{lem:augmenting} Let $P$ be any augmenting path from $X$ to $V$ that unmatches surrogate $\vec{s_i}'$. Let also $S$ denote the set of surrogates in $P$. {Finally, let $\vec{s_i}''$ be the surrogate closest to the end of $P$ opposite to $\vec{s_i}'$. If $\vec{s_i}''$  is matched in both $X$ and $V$, then:
{
$${\frac{\delta+\epsilon}{\eta}}\sum_{\vec{s_i} \in S,j} \pi_{ij}(\vec{s_i}) \geq  p(\vec{s_i}');$$
otherwise
$${\frac{\delta+\epsilon}{\eta}}\sum_{\vec{s_i} \in S,j} \pi_{ij}(\vec{s_i}) \geq p(\vec{s_i}')-p(\vec{s_i}'').$$}}
\end{lemma}

\begin{proof} { Consider the ``weight'' of this path, computed by adding the weights of the new edges and subtracting the weights of the deleted edges. Because VCG maximizes social welfare, this weight must be positive. Decompose the gain/loss in social welfare in two components: that coming from rebates and that coming without accounting for rebates, i.e. were we running mechanism $M_1$ instead of $M$. As far as the first component goes, the contribution to the welfare via rebates from a surrogate that is matched in both $X$ and $V$ does not change in the two matchings. Instead there is a loss of $\eta p(\vec{s_i}')$ of rebates-welfare because $\vec{s_i}'$ becomes unmatched and a gain of $\eta p(\vec{s_i}'')$ of rebates-welfare, if $\vec{s_i}''$ was not matched in $X$ and became matched in $V$.

Now let us upper-bound the gain in social welfare contributed by the $M_1$ component of the edge-weights. Because $M_1$ is $\epsilon$-BIC for ${\cal D}'$, a replica gains at most ${(\delta+\epsilon)} \sum_j \pi_{ij}(\vec{s_i})$ by being matched to $\vec{s_i}$ instead of her equivalent surrogate.  So  the $M_1$-welfare goes up by at most ${(\delta+\epsilon)} \sum_{\vec{s_i} \in S,j} \pi_{ij}(\vec{s_i})$ from the augmentation.

Given that the augmentaion must increase social welfare, we obtain the lemma.}
\end{proof}

This lemma, in essence, says that each time a surrogate becomes unmatched from $X$ to $V$ it ``claims'' some weight of the $\pi_{ij}$s. If we let $W_i$ denote $\sum_{\forall \vec{s_i},j} \pi_{ij}{(\vec{s_i})}$, then we get the following corollary:

{
\begin{corollary} \label{cor:augmenting}$$\sum_{\vec{s_i} \in V} p(\vec{s_i})/r \geq \sum_{\vec{s_i} \in X} p(\vec{s_i})/r - {{\frac{\delta+\epsilon}{\eta r}} W_i}.$$
\end{corollary}}

We now proceed to bound $\sum_i \mathbb{E}[W_i]$ with an easy lemma.

\begin{lemma} \label{lem:boundd} $$\sum_i \mathbb{E}[W_i]/r \leq T.$$
\end{lemma}

\begin{proof}
$\mathbb{E}[W_i]/r$ is exactly the expected number of items awarded to bidder $i$ by $M$. As $M$ can only award $T$ items, we have the desired inequality.
\end{proof}

Now we just have to put everything together and chase through some inequalities. From the work above we get that:

$$\mathbb{E}\left[\sum_{\vec{s_i} \in V} p(\vec{s_i})/r\right] \geq \mathbb{E}\left[\sum_{\forall \vec{s_i}} p(\vec{s_i})/r\right] - \left(T_i\sqrt{\beta \over r} + {  {\frac{\delta+\epsilon}{\eta r} }\mathbb{E}[W_i]}\right).$$
And when we sum this over all bidders we get:

$$R^{M_2}(\mathcal{D}) \geq R^M(\mathcal{D'}) - \sum_i \left(T_i\sqrt{\beta \over r} + { {\frac{\delta + \epsilon}{\eta r} }\mathbb{E}[W_i]}\right)$$
{$$\implies R^{M_2}(\mathcal{D}) \geq R^M(\mathcal{D'}) - \left(\sqrt{\frac{\beta}{r}} \sum_i T_i + \frac{\delta+\epsilon}{\eta}T\right)$$}
Recall that $M$ is just $M_1$ with rebates, hence: 
$$R^M(\mathcal{D'}) = \left(1-\eta\right)R^{M_1}(\mathcal{D'}).$$

Putting the above together with our choice of $r$ {and observing that $T_i \leq T$ for all $i$,} we conclude the proof of Theorem~\ref{thm:epsilon-BIC to BIC}.

\section{Proof of Corollary~\ref{cor:MHR}}\label{app:MHR}

Recall the definition of a Monotone Hazard Rate distribution.

\begin{definition}(Monotone Hazard Rate) A one-dimensional differentiable distribution $F$ satisfies the Monotone Hazard Rate Condition if $\frac{f(x)}{1-F(x)}$ is monotonically non-decreasing for all $x$ such that $F(x) < 1,$ where $f = F'$ is the probability density function.
\end{definition}

\noindent To prove Corollary~\ref{cor:MHR} we reduce the MHR case to the $[0,1]$ case. We do this by finding an appropriate $\Xi$ such that the probability that any bidder values any item above $\Xi$ is tiny. In fact, so tiny that even if we assumed the optimal mechanism could somehow extract full value from bidders when they value an item above $\Xi$, this would account for a tiny fraction of the total revenue. We make use of the following two lemmas from Cai and Daskalakis~\cite{CD}. For a distribution $F$, let $\alpha_p = \inf\{x|F(x)\geq 1-1/p\}$. Then:

\begin{lemma}[\cite{CD}]\label{lem:CD1} If $F$ is MHR, then $k\alpha_p \geq \alpha_{p^k}$, for all $p,k \geq 1$.
\end{lemma}

\begin{lemma}[\cite{CD}] \label{lem:CD2} If $F$ is MHR and $X$ is a random variable distributed according to $F$, then $\mathbb{E}[X|X \geq \alpha_p]\cdot Pr[X \geq \alpha_p] \in O(\alpha_p/p)$.
\end{lemma}

Set $\zeta= \lceil \log_2 k/\epsilon \rceil +1$. For $k$-bidders, let $\alpha_{i,n}$ denote $\alpha_n$ for $F = \mathcal{F}_i$, for all $i$. Let then $\Xi = \max_i \{\alpha_{i,n^{\zeta}}\}$. Then by Lemma \ref{lem:CD2}, even if we could extract full value from every bidder for each item they valued above $\Xi$, we would only make at most $O(k n \Xi/n^\zeta) = O(\epsilon \Xi)$ expected revenue. In addition, there is a trivial mechanism that makes expected revenue $\Omega(\Xi / \log{k/\epsilon})$. Just price each item at $\Xi' = \max_i\{ \alpha_{i,n}\}$ and sell them on a first-come first-served basis. By Lemma \ref{lem:CD1}, we know that $\alpha_{i,n} \geq \alpha_{i,n^{\zeta}}/\zeta$, and therefore $\Xi' \geq \Xi/\zeta$. In addition, the probability that an item gets sold is a constant (approximately $1/e$), so we make revenue $\Omega(\Xi') = \Omega(\Xi/ \log k/\epsilon)$. These two observations together tell us that we can completely ignore the revenue from extreme bidders without losing too much. So our algorithm for $k$-bidders on MHR distributions is as follows: For each $\mathcal{D}_i$, create a new distribution $\mathcal{D}'_i$ that rounds each $v_{ij}$ down to the nearest multiple of $\delta \Xi$ if $v_{ij} < \Xi$, or down to $\Xi$ if $v_{ij} > \Xi$. Find an optimal mechanism $M_1$ for the distribution $\times_i \mathcal{D}'_i$. Then sample surrogates from $\mathcal{D}'_i$ and replicas from $\mathcal{D}_i$ and go through the same reduction as for $[0,1]$ (described in Section~\ref{sec:true PTAS}). Because we sample replicas directly from $\mathcal{D}_i$, this solution will still be BIC (see Appendix \ref{app:true PTAS}). In addition, by the arguments in Appendix \ref{app:true PTAS} and the observations above, we get an additive $O(\epsilon \Xi)$ approximation to the optimal revenue. Because there is a trivial mechanism making revenue $\Omega(\Xi/\log k/\epsilon)$, this is in fact a multiplicative $(1-O(\epsilon \cdot \log k/\epsilon))$ approximation.

For $k$-items, let $\alpha_{j,m}$ denote $\alpha_m$ when $F$ is the marginal distribution of $\mathcal{F}$ for item $j$. Then let $\Xi = \max_j \{\alpha_{j,m^\zeta}\}$. Then by Lemma \ref{lem:CD2}, even if we could extract full value from every bidder for each item they valued above $\Xi$, we could only make at most $O(km\Xi/m^\zeta)=O(\epsilon \Xi)$ expected revenue. In addition, the first-come first-served mechanism that prices all items at $\Xi' = \max_j\{ \alpha_{j,m}\}$ makes $\Omega(\Xi/\log k/\epsilon)$ revenue for the same reasons that the first-come first-served mechanism in the previous paragraph made $\Omega(\Xi/\log k/\epsilon)$ revenue. Again, these observations tell us that we can completely ignore the revenue from extreme bidders without losing too much. Our alogrithm for $k$-items on MHR distributions is the same as for $k$-bidders: $\mathcal{D'}$ samples from $\mathcal{D}$ then rounds each $v_{ij}$ down to the nearest multiple of $\delta \Xi$ if $v_{ij} < \Xi$, and down to $\Xi$ otherwise. We compute the optimal mechanism for ${\cal D}'$. We then sample surrogates from $\mathcal{D}'_i$ and replicas from $\mathcal{D}_i$ and go through the same reduction as for $[0,1]$ (described in Section~\ref{sec:true PTAS}). Again, because we sample replicas directly from $\mathcal{D}$, the solution is still BIC, and by the arguments in Appendix \ref{app:true PTAS} and the observations above, we get an additive $O(\epsilon \Xi)$ approximation to the optimal revenue, which is again a multiplicative $(1-O(\epsilon \cdot \log k/\epsilon))$ approximation.

\section{Extending Theorem~\ref{thm:additive} and Corollary~\ref{cor:MHR} to IC Mechanisms} \label{sec: IC results}

As the modifications to the na\"ive LP for going from BIC to IC are trivial, we will not restate the na\"ive LP for IC here. The symmetry theorem (Theorem \ref{thm:symmetries}, Section \ref{sec:symmetries}) has already been proven for IC mechanisms. The first stop along our proof where we have to treat IC and BIC differently is the monotonicity of item-symmetric mechanisms {(for the $k$-bidders problem)}.

\begin{definition}(Strong-Monotonicity of an IC mechanism) An IC or $\epsilon$-IC mechanism is said to be strongly monotone if for all $i,j,j'$ and $\vec{v}$ such that $v_{kj} = v_{kj'}$ for all $k \neq i$, $\phi_{ij}(\vec{v}) > \phi_{ij'}(\vec{v}) \Rightarrow v_{ij} \geq v_{ij'}$.
\end{definition}

\begin{theorem}\label{thm:ICmonotone} Any item-symmetric IC mechanism is strongly monotone. For all item-symmetric $\epsilon$-IC mechanisms $M$, there exists a mechanism $M'$ of equal revenue that is strongly monotone.
\end{theorem}

\begin{proof} The proof follows the same lines as that of Theorem \ref{thm:monotone} after making a quick observation. If $v_{kj} = v_{kj'}$ for all $k \neq i$ and $\sigma$ is the permutation that swaps items $j$ and $j'$, then if bidder $i$ swaps $v_{ij}$ and $v_{ij'}$, he turns $\vec{v}$ into $\sigma(\vec{v})$, simply because $\sigma$ does not affect $\vec{v}_{-i}$. As any item-symmetric mechanism $M$ must have $M(\sigma(\vec{v})) = \sigma(M(\vec{v}))$, the rest of the proof follows that of Theorem \ref{thm:monotone} as bidder $i$ can swap the probability that he receives items $j$ and $j'$ by swapping his values.
\end{proof}

Next, we have to turn Theorem \ref{thm:ICmonotone} into monotonicity constraints for the LP {of the $k$-bidders problem}. We say that $\vec{v} \sim_i \vec{w}$ if there exists a $\sigma$ such that $\sigma(\vec{v}) = \sigma(\vec{w})$ and $\sigma(\vec{v}_k) = \vec{v}_k$ for all $k \neq i$. In other words, $\vec{w}_i$ is a permutation of $\vec{v}_i$ that is the identity on $\vec{v}_{-i}$. Then let $E_i(\vec{v}_{-i})$ denote the set of $\vec{w}$ such that {($\vec{w}_{-i} = \vec{v}_{-i}$) and ($\forall j \wedge j < j': w_{kj} = w_{kj'}  \Rightarrow w_{ij} \geq w_{ij'}$)}. I.e. {$E_i(\vec{v}_{-i})$} is a set of representatives under the equivalence relation $\sim_i$ for a fixed $\vec{v}_{-i}$. Strong-monotonicity implies then the following.

\begin{observation}\label{obs: ICmonotone} When playing an {item-symmetric}, strongly monotone IC mechanism, if $v_{kj} = v_{kj'}$ for all $k\neq i$, and $v_{ij} {>} v_{ij'}$, then bidder $i$ has no incentive report {any $\vec{w}_i$ such that $w_{ij'} > w_{ij}$.}
\end{observation}

\begin{corollary}\label{cor: ICmonotone} If {$M$ is strongly-monotone and item-symmetric} and when playing $M$, for all $i$ and $\vec{v}_{-i}$, bidder $i$ never has (more than $\epsilon$) incentive to misreport $\vec{w}_i \in E_i(\vec{v}_{-i})$ when her true type is $\vec{v}_i \in E_i(\vec{v}_{-i})$ for any $\vec{w}_i,\vec{v}_i$, then $M$ is IC ($\epsilon$-IC).
\end{corollary}

{Given Corollary~\ref{cor: ICmonotone}, we replace in the $k$-bidders LP the BIC and strong-monotonicity constraints with the following IC and strong-monotonicity constraints:}

\paragraph{IC, Strong-Monotonicity Constraints:}
$$\sum_j {v_{ij}}\phi_{ij}(\vec{v}) - p_i(\vec{v}) \geq \sum_j {v_{ij}}\phi_{ij}(\vec{w}) - p_i(\vec{w}),~~~\text{for all $i,\vec{v}\in E,\vec{w} \in E_i(\vec{v}_{-i})$;}$$
$$\phi_{ij}(\vec{v}) \ge \phi_{ij'}(\vec{v}),~~~\text{for all $i$, $j < j'$ and $\vec{v}\in E$ such that $v_{kj} = v_{kj'}$ for all $k \neq i$}.$$

\medskip {The last step is to show that there are only polynomially many IC constraints.} We show this in the following lemma:

\begin{lemma}\label{lem: numberIC} $|E_i(\vec{v}_{-i})| \leq |E|$ for all $i,\vec{v}_{-i}$.
\end{lemma}
\begin{proof}
We prove the lemma by showing that if $\vec{w} \in E_i(\vec{v}_{-i})$, and $\sigma(\vec{w}) \in E_i(\vec{v}_{-i})$, then $\sigma(\vec{w}) = \vec{w}$. Observe first that in order to possibly have $\sigma(\vec{w}) \in E_i(\vec{v}_{-i})$, we must have $\sigma(\vec{w}_{-i}) = \vec{v}_{-i}=\vec{w}_{-i}$. In other words, if $\sigma(j) = j'$, then $w_{kj} = w_{kj'}$ for all $k \neq i$. However, in order for $\vec{w} \in E_i(\vec{v}_{-i})$, it must be the case that for all such $j,j'$, $w_{ij} > w_{ij'} \Rightarrow j < j'$, which means that there is a unique ordering of such values in $\vec{w}_i$ that will make $\vec{w} \in E_i(\vec{v}_{-i})$. Because $\vec{w}$ and $\sigma(\vec{w})$ must respect the same ordering, they must be the same vector.

Therefore, because $E_i(\vec{v}_{-i})$ contains at most one representative per equivalence class under $\sim$, and $E$ contains exactly one representative per equivalence class, we have that $|E_i(\vec{v}_{-i})| \leq |E|$.
\end{proof}

Our discretization lemma of Section~\ref{sec:epsilon-truthful}, namely Lemma~\ref{lem:deltaIC}, has exactly the same statement, replacing BIC with IC, and its proof is very similar (just removing expectations over the other bidders' types where appropriate). As we do not have an $\epsilon$-IC to IC reduction, the above changes are actually the only changes that are needed to adjust Theorem \ref{thm:additive} to $\epsilon$-IC.  The proof of Corollary \ref{cor:MHR} also works for $\epsilon$-IC, so we have shown how to obtain $\epsilon$-IC mechanisms for the settings of Theorem~\ref{thm:additive} and Corollary~\ref{cor:MHR}. As discussed in Remark \ref{rem:additive}, we can accommodate arbitrary budget constraints as we do not use an analogue of the $\epsilon$-BIC to BIC reduction, which was the only step in our BIC proof that could not accommodate budgets.

\end{document}